\documentclass[12pt]{article}

\usepackage[x11names]{xcolor}
\usepackage{tikz}
\usepackage{cite}
\usepackage{comment}
\usepackage{hyperref}
\usetikzlibrary{decorations.pathreplacing,decorations.markings}
\usepgflibrary{arrows}
\usepgflibrary[arrows]
\usetikzlibrary{arrows}
\usetikzlibrary[arrows]
\usetikzlibrary{positioning}

\usepackage{amsmath}
\usepackage{amssymb}
\usepackage{subsupscripts}
\usepackage{braket}
\usepackage{hyperref}

\usepackage{amsthm}
\theoremstyle{plain}
\newtheorem{thm}{Theorem}[section]
\newtheorem{prop}[thm]{Proposition}
\newtheorem{lem}[thm]{Lemma}
\newtheorem{cor}[thm]{Corollary}

\theoremstyle{definition}
\newtheorem{defn}[thm]{Definition}

\tikzset{
  on each segment/.style={
    decorate,
    decoration={
      show path construction,
      moveto code={},
      lineto code={
        \path [#1]
        (\tikzinputsegmentfirst) -- (\tikzinputsegmentlast);
      },
      curveto code={
        \path [#1] (\tikzinputsegmentfirst)
        .. controls
        (\tikzinputsegmentsupporta) and (\tikzinputsegmentsupportb)
        ..
        (\tikzinputsegmentlast);
      },
      closepath code={
        \path [#1]
        (\tikzinputsegmentfirst) -- (\tikzinputsegmentlast);
      },
    },
  },
  mid arrow/.style={postaction={decorate,decoration={
        markings,
        mark=at position .6 with {\arrow[#1]{triangle 45}}
      }}},
}

\newdimen\shadedBaseline\shadedBaseline=-4mm
\newcount\tableauRow\newcount\tableauCol
\newcommand\ShadedTableau[2][\relax]{%
  \begin{tikzpicture}[scale=0.3,draw/.append style={thick,black},baseline=\shadedBaseline]
    \ifx\relax#1\relax%
    \else 
      \foreach\bx in {#1} { \filldraw[blue!20]\bx+(-.5,-.5)rectangle++(.5,.5); }
    \fi
    \tableauRow=0
    \foreach \Row in {#2} {
       \tableauCol=1
       \foreach\k in \Row {
          \draw(\the\tableauCol,\the\tableauRow)+(-.5,-.5)rectangle++(.5,.5);
          \draw(\the\tableauCol,\the\tableauRow)node{\k};
          \global\advance\tableauCol by 1
       }
       \global\advance\tableauRow by -1
    }
  \end{tikzpicture}%
}

\makeatletter
\@addtoreset{equation}{section}

\makeatother

\begin{document}
\begin{titlepage}
\renewcommand{\thefootnote}{\fnsymbol{footnote}}
\vspace*{-2cm}
\begin{flushright}
UT-19-05
\end{flushright}

\vspace*{1cm}

\begin{center}
{\huge A note on S-dual basis in free fermion system}
\par
\vspace{1cm}
{\Large Shinya Sasa, Akimi Watanabe and Yutaka Matsuo}
\\[.6cm]
{\it Department of Physics, The University of Tokyo}\\
{\it 7-3-1 Hongo, Bunkyo-ku, Tokyo 113-0033, Japan}
\\[.4cm]
\texttt{E-mail: sasa, awatanabe, matsuo at hep-th.phys.s.u-tokyo.ac.jp}

\end{center}

\vspace{3cm}

\begin{abstract}
\noindent

Free fermion system is the simplest quantum field theory which has the symmetry of Ding-Iohara-Miki algebra (DIM).  DIM has S-duality symmetry, known as Miki automorphism which defines the transformation of generators. 
In this note, we introduce the second set of the fermionic basis (S-dual basis) which implement the duality transformation. It may be interpreted as the Fourier dual of the standard basis, and the inner product between the standard and the S-dual ones is proportional to the Hopf link invariant. We also rewrite the general topological vertex in the form of Awata-Feigin-Shiraishi intertwiner and show that it becomes more symmetric for the duality transformation.
\vspace{0.5cm}
\end{abstract}

\vfill

\end{titlepage}
\vfil\eject

\setcounter{footnote}{0}

\section{Introduction}

The free fermion system in two dimensions is one of the simplest examples in the quantum field theory and has applications in many branches of physics and mathematics, such as the topological string, the quantum Hall effect, and the quantum gravity in two dimensions.
It is known to have (quantum) $W_{1+\infty}$ algebra symmetry \cite{Bakas:1989xu, Pope:1989ew} which plays a substantial role to understand these systems. The representation theory of the symmetry was studied \cite{Frenkel:1994em, Awata:1994tf} and it shows that it describes only the free systems.

A few years later, it was generalized to a deformed symmetry which describes the interacting system. It was studied by some groups independently and has different names, such as Ding-Iohara-Miki (DIM) algebra \cite{Ding:1996, Miki:2007}, the quantum continuous gl($\infty$) \cite{Feigin:2010}, the quantum toroidal algebra \cite{feigin2012quantum}. In this paper, we refer to it as DIM algebra.
Recently, Schiffmann et al. \cite{Schiffmann:2012} used the symmetry to prove Alday-Gaiotto-Tachikawa (AGT) conjecture \cite{Alday:2009aq}, the equivalence between the conformal block function of two-dimensional CFT and the instanton partition function of $N=2$ super Yang-Mills. While the former has the expression by the basis of the Virasoro Hilbert space, the latter is labeled by Young diagrams which represent the fixed point of the localization. The equivalence between the two implies the algebra has nontrivial dual descriptions.

The generators of the DIM  has a label in $\mathbb{Z}^2$ and $SL(2,\mathbb{Z})$ acts on it as an automorphism.  In particular, the S-duality action is called Miki automorphism. Among the generators of DIM, we may choose two sets of generators which are related by Miki automorphism. In one set, there is a free boson realization of generators, which is the standard realization of the CFT.  On the other hand, in terms of the other set of generators, the representation is given by the basis labeled by Young diagrams, which appear in Nekrasov partition function.

The purpose of this note is to present an alternative picture of S-duality.  Instead of changing generators, we introduce the second set of basis which implements the dual descriptions of the CFT. The existence of two pictures is helpful to understand the nature of the duality. In order to simplify the setup, we restrict ourselves to the free fermion system where the DIM algebra reduces to the quantum $W_{1+\infty}$ which is isomorphic to the quantum torus algebra up to the central charge. At the first quantized level, the duality map on the states is the Fourier transformation. At the second quantized level, we have to be careful to define the states to avoid the divergence. In order to show that the action on the S-dual basis is the representation rotated by right-angle, we need to introduce a projection operator which is necessary to modify the central charge together with shifts of generators.

The DIM algebra also appears in the topological string where one may identify the topological vertex \cite{Aganagic:2003db} as an intertwiner between the representations \cite{Awata:2011ce}.  In this context, it is better to interpret the $\mathbb{Z}^2$ charges of the representation with the brane charges. We first show that the inner product between the standard and S-dual bases is proportional to Hopf link invariant. We also show that the intertwiner takes more symmetric form if we use the S-dual basis. In particular, the S-duality in the topological string amplitude becomes a consequence of the symmetry between the standard and the S-dual bases.

We note that the duality of the DIM algebra was studied in various contexts. For the recent studies, see for instance \cite{Awata:2011dc, Bourgine:2018fjy, Fukuda:2019ywe}.

This paper is organized as follows. In section \ref{sec:reviews}, we give a brief review on the DIM algebra and its representation, starting from the quantum torus. This part also serves as a pedagogical introduction to the duality of the quantum $W_{1+\infty}$ algebra. In section \ref{sec:S-dual}, we define the S-dual basis and derive S-dual transformation between the $(1,0)$ and $(0,1)$ representations using them. We also explain the nature of the projection operator, which is necessary to reproduce the representation. In section \ref{s:tv}, we compute the inner product between the two distinguished bases. We also derive the expression of the intertwiner and calculate the amplitude. Finally, we present our conclusion in section \ref{sec:conclusion}.

\section{Preliminaries} \label{sec:reviews}
\subsection{Quantum torus}
It is illuminative to start the explanation of the origin of the S-duality automorphism in the DIM algebra from the quantum torus. This algebra is generated by $U$ and $V$  which satisfy the commutation relation
\begin{equation}
    UV = qVU.
\end{equation}
In the quantum mechanical system, $U$ and $V$ appear as translation operators along the edge of a rectangle in the presence of constant magnetic flux. One may write the general element of the translation as $U^m V^m$. In the following, we consider the generic case where $q$ is not the root of unity.
The commutation relation among them is 
\begin{equation}
    [U^{m_1} V^{n_1}, U^{m_2} V^{n_2}] = (q^{-m_2n_1} - q^{-m_1n_2}) U^{m_1+m_2} V^{n_1+n_2}.
\end{equation}
We call this algebra as the quantum torus algebra. As the name suggests,
it has an $SL(2, \mathbb{Z})$ automorphism,
\begin{gather}
	\mathcal{T} = \begin{pmatrix}
		a & c \\
        b & d
	\end{pmatrix}\in SL(2,\mathbb{Z}),
\end{gather}
\begin{gather}
	U \to U^\mathcal{T} = U^a V^b, \quad V \to V^\mathcal{T} = U^c V^d,\quad
	U^\mathcal{T}V^\mathcal{T} =q V^\mathcal{T} U^\mathcal{T}.
\end{gather}
Among the $SL(2,\mathbb{Z})$ transformation,  $\mathcal{S}$-transformation
\begin{gather}
	\mathcal{S} = \begin{pmatrix}
		0 & -1 \\
        1 & 0
	\end{pmatrix}, \quad U \to U^\mathcal{S} = V, \quad V \to V^\mathcal{S} = U^{-1} \label{eq:qtorus_S}
\end{gather}
will play a special role in the following.

One may represent the quantum torus algebra by multiplication of $z$ and the $q$-difference operator:
\begin{gather}
	U = q^D = q^{z \frac{d}{dz}}, \qquad V = z. \label{eq:qtorus_vector}
\end{gather}
One choice of the basis $\{f_n\}_{n\in \mathbb{Z}}$ of the representation space is
\begin{equation}
    f_n(z) = z^n,\quad n\in \mathbb{Z},
\end{equation}
where the action of the $U$ operator is diagonal,
\begin{gather}
	\quad Uf_n(z) = q^n f_n(z), \\
	\quad Vf_n(z) = f_{n+1}(z).
\end{gather}
It is referred to as an (infinite) vector representation.
We can define the second set of the basis which is diagonal to the $V$ generators, using the delta function $\delta(z) = \sum_{n \in \mathbb{Z}} z^n$, 
\begin{equation}\label{e:FBasis}
    f_n^\mathcal{S}(z) = \delta(z/q^n)=\sum_{m\in \mathbb{Z}} z^m q^{-nm},\quad
    n\in \mathbb{Z}.
\end{equation}
$f_n^\mathcal{S}(z) $ is the Fourier transformation of the standard basis.
The $\mathcal{S}$-transformed generators have the same action on them.
\begin{gather}
	\quad U^\mathcal{S} f_n^\mathcal{S}(z) = V f_n^\mathcal{S}(z) =  q^nf_n^\mathcal{S}(z), \\
	\quad V^\mathcal{S} f_n^\mathcal{S}(z) = U^{-1} f_n^\mathcal{S}(z) = f_{n+1}^\mathcal{S}(z). 
\end{gather}
We call them as the `S-dual' basis of the quantum torus algebra.
In the following, we study the second quantized version.

\subsection{Ding-Iohara-Miki algebra}
\begin{defn}
    The DIM algebra is a quantum algebra with three parameters $q_1$, $q_2$, $q_3 \in \mathbb{C}$ satisfying $q_1 q_2 q_3 = 1$. They are related with those of the Macdonald polynomials or the refined topological vertex as $q_1=q$, $q_2=t^{-1}$.
    The algebra is generated by $\left\lbrace x_i^{\pm}, \psi_{\pm j}^{\pm}, \hat{\gamma}^{\pm 1/2} | i \in \mathbb{Z}, j \in \mathbb{Z}_{\geq 0} \right\rbrace$. The generating function is referred to as Drinfeld currents,
    \begin{align}
    	&x^\pm (z) = \sum_{k \in \mathbb{Z}} z^{-k} x^\pm_k, \\
        &\psi^+ (z) = \sum_{k \geq 0} z^{-k} \psi^+_k, \\
        &\psi^- (z) = \sum_{k \geq 0} z^{k} \psi^-_{-k}\,.
    \end{align}
$\hat{\gamma}^{\pm 1/2}$ is the center. 
The DIM algebra is written as (see for instance \cite{tsymbaliuk2017affine})
    \begin{align}
        & \psi^+_0 \psi^-_0 = \psi^-_0 \psi^+_0 = 1, \\
    	&[\psi^\pm (z), \psi^\pm (w)] = 0, \\ 
        &\psi^+(z) \psi^- (w) = \frac{g(\hat{\gamma}w/z)}{g(\hat{\gamma}^{-1}w/z)} \psi^-(w) \psi^+(z), \\
        & \psi^+(z) x^\pm (w) = g(\hat{\gamma}^{\mp 1/2}w/z)^{\mp 1} x^\pm (w) \psi^+ (z), \\
        & \psi^-(z) x^\pm (w) = g(\hat{\gamma}^{\mp 1/2}z/w)^{\pm 1} x^\pm (w) \psi^- (z), \\
        & x^\pm(z) x^\pm(w) = g(z/w)^{\pm 1} x^\pm(w) x^\pm(z), \\
        & [x^+(z),x^-(w)] = \frac{(1-q_1)(1-q_2)}{1-q_1q_2} \left( \delta(\hat{\gamma}^{-1}z/w) \psi^+(\hat{\gamma}^{1/2}w) - \delta(\hat{\gamma}z/w) \psi^- (\hat{\gamma}^{-1/2}w) \right),
    \end{align}
    where
    \begin{equation}
        g(z) = \prod_{\alpha=1, 2, 3} \frac{1-q_\alpha z}{1-q_\alpha^{-1}z},
    \end{equation}
together with Serre relations.
\end{defn}

The DIM algebra has an automorphism found by Miki~\cite{Miki:2007}:
\begin{equation}
    \begin{minipage}{0.3\hsize}
        \centering
        \begin{tikzpicture}
            \node (1) at (0,1) {$h_1$};
            \node (2) at (1,0) {$e_0$};
            \node (3) at (0,-1) {$h_{-1}$};
            \node (4) at (-1,0) {$f_0$};
            \draw[->] (1) -- (2);
            \draw[->] (2) -- (3);
            \draw[->] (3) -- (4);
            \draw[->] (4) -- (1);
        \end{tikzpicture}
    \end{minipage}
    \begin{minipage}{0.3\hsize}
        \centering
        \begin{tikzpicture}
            \node (1) at (0,1) {$\psi^+_0$};
            \node (2) at (1,0) {$\hat{\gamma}^{-1}$};
            \node (3) at (0,-1) {$\psi^-_0$};
            \node (4) at (-1,0) {$\hat{\gamma}$};
            \draw[->] (1) -- (2);
            \draw[->] (2) -- (3);
            \draw[->] (3) -- (4);
            \draw[->] (4) -- (1);
        \end{tikzpicture}
    \end{minipage}.
\end{equation}
Here we change the normalization of the generators as
\begin{align}
    e_0 &= - \frac{x^+_0}{q_1^{-1/2}q_3^{1/4}(1-q_1)^2}, \\
    f_0 &= \frac{x^-_0}{q_2^{-1/2}q_3^{1/4}(1-q_2)^2}.
\end{align}
We note that DIM is generated by four generators $h_{\pm 1}, e_0, f_0$ and the map here extends to all generators of DIM.
This automorphism corresponds to the $\mathcal{S}$-transformation of the quantum torus (\ref{eq:qtorus_S}).

In this paper, we focus on the self-dual case $q=t$. We define the generators $h_m \ (m \in \mathbb{Z}_{\neq 0})$, $\hat{l}_1$, $\hat{l}_2$ by the relations
\begin{gather}
    \psi^\pm ( \hat{\gamma}^{1/2} z) = \psi^\pm_0 \exp \left( \sum_{m>0} \kappa_m h_{\pm m} z^{\mp m} \right), \\
    \hat{\gamma}^{\pm 1/2} = q_3^{\pm\hat{l}_1/4}, \qquad \psi^\pm_0 = q_3^{\mp \hat{l}_2/2},
\end{gather}
where $\kappa_m = (1-q_1^m)(1-q_2^m)(1-q_3^m)$. Then commutation relations can be rewritten as~\cite{Harada:2018} 
\begin{align}
    [h_m, x^+_n] &= - \frac{1}{m} x^+_{n+m}, \\
    [h_m, x^-_n] &= \frac{1}{m} x^-_{n+m}, \\
    [h_m, h_n] &= - \delta_{m+n,0} \frac{1}{m} \frac{\hat{l}_1}{(1-q^m)(1-q^{-m})}, \\
    [x^+_m, x^-_n] &= - (1-q)(1-q^{-1}) \left[ (m+n)(1-q^{m+n})(1-q^{-m-n})h_{m+n} + \delta_{m+n,0} (n\hat{l}_1 + \hat{l}_2) \right].
\end{align}
We assume that $|q|<1$ and $q \neq 0$ in the rest of this paper.

The DIM algebra has a family of representations labeled by a weight $u$ and a level $(l_1, l_2) \in \mathbb{Z}^2$, which is the eigenvalues of $\hat{l}_1$ and $\hat{l}_2$. In this paper, we use only the vertical representation $(0,1)_u$ and the horizontal representation $(1,n)_u$. In the self-dual limit, Miki automorphism reduces to the transformation $\hat{l}_1 \to -\hat{l}_2$ and $\hat{l}_2 \to \hat{l}_1$, thus the levels of the representations also transform as $l_1 \to -l_2$ and $l_2 \to l_1$. Especially, the $(1,0)$ representation is mapped to $(0,1)$.

\subsection{DIM algebra and the quantum torus}

Here we introduce the free bosonic field
\begin{equation}
	\phi(z) = \hat{q} + a_0 \log(z) - \sum_{n \neq 0} \frac{a_n}{n} z^{-n},
\end{equation}
\begin{gather}
    [a_m, a_n] = m \delta_{m+n, 0}, \qquad [a_0, \hat{q}] = 1 \label{eq:boson_CR}
\end{gather}
and the free fermionic field
\begin{gather}
    \psi(z) = \sum_{r \in \mathbb{Z}+1/2} \psi_r z^{-r-1/2}, \qquad \psi^\dagger(z) = \sum_{r \in \mathbb{Z}+1/2} \psi^\dagger_r z^{-r-1/2},
\end{gather}
\begin{gather}
    \lbrace \psi_r, \psi^\dagger_s \rbrace = \delta_{r+s, 0}, \qquad \lbrace \psi_r, \psi_s \rbrace = \lbrace \psi^\dagger_r. \psi^\dagger_s \rbrace = 0. \label{eq:fermion_AR}
\end{gather}

The self-dual limit of the DIM algebra is also obtained by 2-dimensional central extension of the quantum torus. We can realize this central extension by the free fields as
\begin{align}
	W[U^mV^n] &= \oint \frac{dz}{2\pi i} \colon e^{-\phi(z)} \colon q^{mD}z^n \colon e^{\phi(z)} \colon = \frac{1}{1-q^m}\oint \frac{dz}{2\pi i z} (q^mz)^n \colon e^{-\phi(z) + \phi(q^mz)} \colon \notag \\
    &= \oint \frac{dz}{2\pi i} \psi(z) q^{mD} z^n \psi^\dagger(z) = \sum_{r \in \mathbb{Z}+1/2} q^{m(r-1/2)} \colon \psi_r \psi^\dagger_{n-r} \colon + \frac{\delta_{n,0}}{1-q^m}
\end{align}
when $m \neq 0$. The normal ordering $::$ is defined by bosonic (fermionic) modes depending on the oscillators appearing in the formula. In the $m=0$ case, we define
\begin{equation}
    W[V^n] = -a_n = \sum_{r \in \mathbb{Z} + 1/2} \colon \psi_{r} \psi^\dagger_{n-r} \colon .
\end{equation}
Representations of the DIM algebra can be obtained by corresponding these $W$'s with the generators.

For later use, we introduce the commutative bosons as
\begin{equation}
    b_n =  W[U^{-n}] = -\sum_{r \in \mathbb{Z} + 1/2} q^{n(r+1/2)} \colon \psi^\dagger_r \psi_{-r} \colon + \frac{1}{1-q^{-n}} \qquad (n \in \mathbb{Z}_{\neq 0}).
\end{equation}
The relation between the quantum torus algebra and DIM may be more explicitly seen by introducing
\begin{equation}
	\psi_f = \oint \frac{dz}{2\pi i} \psi(z) f(z), \quad \psi^\dagger_f = \oint \frac{dz}{2\pi i} \psi^\dagger(z) f(z).
\end{equation}
We note that these give maps from a wave function $f(z)$ into the operator in the free fermion Fock space.
The action of the quantum torus algebra on $f$ is obtained by the commutation relation with $W[U^mV^n]$,
\begin{gather}
	\left[ W[U^mV^n], \psi_f \right] = \psi_{U^mV^nf}\,, \\
    \left[ W[U^mV^n], \psi_f^\dagger \right] = -q^{-m} \psi^\dagger_{V^nU^{-m}f}\,.
\end{gather}

\subsection{Two representations of DIM}
There are two types of representations of DIM, the horizontal and the vertical ones, depending on the central charges $\hat{l}_{1,2}$ \cite{Awata:2011ce}. At the level of the quantum torus, they correspond to the different choice of the generators $U,V$.
\begin{itemize}
\item Horizontal $(1,n)$ representation corresponds to the choice,
\begin{equation}
U^{\mathcal{T}_{-n}}=UV^{-n},\quad
V^{\mathcal{T}_{-n}}=V,\quad
	\mathcal{T}_{-n} = \begin{pmatrix}
		1 & 0 \\
        -n & 1
	\end{pmatrix}\,.
\end{equation}
\item Vertical $(0,1)$ representation corresponds to the choice,
\begin{equation}
U^{\mathcal{S}^{-1}}=V^{-1},\quad
V^{\mathcal{S}^{-1}}=U,\quad
	\mathcal{S}^{-1} = \begin{pmatrix}
		0 & 1 \\
        -1 & 0
	\end{pmatrix}\,.
\end{equation}
\end{itemize}
At the second quantized level, we define the Drinfeld currents by
\begin{eqnarray}
x^+_{\mathcal{R}}(z)=(1-q)W[U^\mathcal{R}\delta(V^\mathcal{R}/z)],\quad
x^-_{\mathcal{R}}(z)=(1-q^{-1})W[\delta(V^\mathcal{R}/z)(U^\mathcal{R})^{-1}]\,,
\end{eqnarray}
where $\mathcal{R}=\mathcal{T}_{-n}$ or $\mathcal{S}^{-1}$.
We use the same Fock space in order to describe two representations.  Because of the choice of the generators, the representations look very different.

\subsubsection{Horizontal $(1,n)_u$ representation}

The horizontal $(1,n)_u$ representation~\cite{Awata:2011ce} are most conveniently written in terms of bosonized currents,
\begin{gather}
	x^+_{\mathcal{T}_{-n}}(z) = (1-q) W [UV^{-n}\delta(V/z)] = q^{a_0} z^{-n} :e^{-\sum_{m \neq 0} \frac{1-q^{m}}{m} a_m z^{-m}}:, \\
	x^-_{\mathcal{T}_{-n}}(z) = (1-q^{-1}) W [\delta(V/z)V^{n}U^{-1}] = q^{-a_0} z^n :e^{\sum_{m \neq 0} \frac{1-q^{m}}{m} a_m z^{-m}}:, \\
	h_m = - \frac{1}{m} \frac{a_m}{1-q^{-m}}, \\
    \hat{l}_1 = 1, \qquad \hat{l}_2 = n.
\end{gather}
We have an extra parameter $u\in \mathbb{C}$ which corresponds to the vacuum charge
in the bosonic Fock space,
\begin{gather}
	\ket{u} = u^{\hat{q}/\log q} \ket{0}, \\
	\qquad a_n \ket{0} = 0 \qquad (n \geq 0).
\end{gather}

\subsubsection{Vertical $(0, 1)_v$ representation}
This representation is realized by the '$\mathcal{S}$-transformed' operators as
\begin{gather}
	x^+_{\mathcal{S}^{-1},v}(z) = (1-q) W [V^{-1} \delta(Uv/z)] = (1-q)\sum_{r \in \mathbb{Z}+1/2} \delta(vq^{r+1/2}/z) \colon \psi_r \psi_{-1-r}^\dagger \colon, \\
    x^-_{\mathcal{S}^{-1},v}(z) = (1-q^{-1}) W [\delta(Uv/z) V] = (1-q^{-1})\sum_{r \in \mathbb{Z}+1/2} \delta(vq^{r-1/2}/z) \colon \psi_r \psi_{1-r}^\dagger \colon, \\
    h_m = \frac{v^m}{m} \frac{b_{-m}}{1-q^{-m}}, \\
    \hat{l}_1 = 0, \qquad \hat{l}_2 = 1.
\end{gather}
As is obvious from these representations, their action is most straightforwardly written in terms of fermionic basis which is labeled by a Young diagram $\lambda$,
\begin{equation}
    \ket{\lambda} = \prod_{i=1}^{\mathrm{diag}(\lambda)} \left[ (-1)^{-\lambda_i+i} \psi_{-\lambda'_i+i-1/2} \psi^\dagger_{-\lambda_i+i-1/2} \right] \ket{0},
\end{equation}
where $\lambda_i$ (resp. $\lambda'_i$) is the length (resp. height) of $i$-th row (resp. column), and $\mathrm{diag}(\lambda)$ is the number of the rows satisfying $i \leq \lambda_i$. We note that the signs of these states are different from the standard convention.

In order to express the action of generators, we introduce some notations.
We associate each box $(i,j)$ in a Young diagram\footnote{The box coordinate is defined to satisfy $(i,j) \in \lambda \Leftrightarrow 1 \leq i \leq l(\lambda), 1 \leq j \leq \lambda_i$.} with the quantity
\begin{equation}
	\chi_{(i,j),v;q} = vq^{-i+j}.
\end{equation}
We sometimes omit the subscripts $v$ and $q$. We introduce a half-integer sequence
$\rho = (-1/2, -3/2, -5/2, \cdots)$ and integer sequences $\varepsilon_i$ saisfying
\begin{equation}
    (\varepsilon_i)_j =
    \begin{cases}
        1 \qquad &(i=j) \\
        0 \qquad &(\text{otherwise})
    \end{cases}.
\end{equation}
The action on the basis is given as~\cite{Awata:2011ce} 
\begin{align}
    x^+_{\mathcal{S}^{-1},v}(z) \ket{\lambda} &= (1-q) \sum_{x \in A(\lambda)} \delta(\chi_{x}/z) \ket{\lambda+x}, \label{eq:x+_Sinverse_action} \\
    x^-_{\mathcal{S}^{-1},v}(z) \ket{\lambda} &= (1-q^{-1}) \sum_{x \in R(\lambda)} \delta(\chi_{x}/z) \ket{\lambda-x}, \label{eq:x-_Sinverse_action} \\
    h_m \ket{\lambda} &= \frac{v^m}{m} \frac{p_{m}(q^{-\lambda'-\rho-1/2})}{1-q^{-m}} \ket{\lambda} \qquad (m>0), \\
    h_m \ket{\lambda} &= - \frac{v^m}{m} \frac{p_{-m}(q^{-\lambda-\rho+1/2})}{1-q^{-m}} \ket{\lambda} \qquad (m<0),
\end{align}
where $p_m(x) = \sum_{i} x_i^m$ is the power sum and $A(\lambda)$ ($R(\lambda)$) is a set of boxes which can be added to (removed from) a Young diagram $\lambda$:
\begin{gather}
    A(\lambda) = \lbrace (i, \lambda_i+1) \mid i \in \mathbb{Z}_{> 0}, \lambda_{i-1} > \lambda_i \rbrace = \lbrace (\lambda'_j+1, j) \mid j \in \mathbb{Z}_{> 0}, \lambda'_{j-1} > \lambda'_j \rbrace, \\
    R(\lambda) = \lbrace (i, \lambda_i) \mid i \in \mathbb{Z}_{> 0}, \lambda_{i+1} < \lambda_i \rbrace = \lbrace (\lambda'_j, j) \mid j \in \mathbb{Z}_{> 0}, \lambda'_{j+1} < \lambda'_i \rbrace.
\end{gather}
We note that $\lambda_i=0$ for $i > l(\lambda)$.

The action of $a_m$ on this basis is given by the Murnaghan-Nakayama rule
\begin{align}
	a_{-m} \ket{\lambda} &= \sum_{\mu \in Y_A^m(\lambda)} (-1)^{\mathrm{ht}(\mu'/\lambda')+1} \ket{\mu}, \label{eq:a-_action} \\
    a_{m} \ket{\lambda} &= \sum_{\mu \in Y_R^m(\lambda)} (-1)^{\mathrm{ht}(\lambda'/\mu')+1} \ket{\mu}, \label{eq:a+_action} \\
    a_0 \ket{\lambda} &= 0, \label{eq:a0_action}
\end{align}
where $m>0$ and $\mathrm{ht}(r)+1$ is the number of the rows occupied by the ribbon\footnote{A ribbon is a connected skew diagram which does not contain $2 \times 2$ blocks~\cite{Macdonald:1979}.} $r$. $Y_A^m(\lambda)$ and $Y_R^m(\lambda)$ are given as
\begin{align}
	Y_A^m(\lambda) &= \Set{ \mu : \textup{partition} | \exists i \in \mathbb{Z}_{>0} \quad \text{s.t.} \quad \mu + \rho = (\lambda + \rho + m\epsilon_i)^+ } \notag \\
	&= \Set{ \mu : \textup{partition} | \exists j \in \mathbb{Z}_{>0} \quad \text{s.t.} \quad \mu' + \rho = (\lambda' + \rho + m\epsilon_j)^+ }, \\
	Y_R^m(\lambda) &= \Set{ \mu : \textup{partition} | \exists i \in \mathbb{Z}_{>0} \quad \text{s.t.} \quad \mu + \rho = (\lambda + \rho - m\epsilon_i)^+ } \notag \\
	&= \Set{ \mu : \textup{partition} | \exists j \in \mathbb{Z}_{>0} \quad \text{s.t.} \quad \mu' + \rho = (\lambda' + \rho - m\epsilon_j)^+ }.
\end{align}
Here $a^+$ denote the sequence obtained by rearranging the terms in $a$ in a descending order. 
Diagrammatically, $Y_A^m(\lambda)$ ($Y_R^m(\lambda)$) is a set of Young diagrams which can be obtained by adding (removing) a ribbon with $m$ boxes to (from) $\lambda$, as shown in Fig.\ref{fig:ribbon}.

\begin{figure}
    \begin{center}
        \begin{minipage}{0.4\hsize}
            \begin{tikzpicture}[scale=0.3]
                \fill[fill=red!50] (6,-5)--(7,-5)--(7,-4)--(10,-4)--(10,-3)--(12,-3)--(12,-2)--(9,-2)--(9,-3)--(6,-3)--cycle; 
                \draw[very thick] (0,0)--(0,-9)--(4,-9)--(4,-6)--(6,-6)--(6,-3)--(9,-3)--(9,-2)--(13,-2)--(13,0)--cycle;
                \node[above] (l) at (3,-3) {$\lambda$};
                \node[left] (i) at (0,-4.4) {$i$};
                \node[above] (j) at (11.5,0) {$j$};
            \end{tikzpicture}
        \end{minipage}
        \begin{minipage}{0.4\hsize}
            \begin{tikzpicture}[scale=0.3]
                \fill[fill=blue!50] (2,-9)--(3,-9)--(4,-9)--(4,-6)--(6,-6)--(6,-4)--(5,-4)--(5,-5)--(3,-5)--(3,-8)--(2,-8)--cycle; 
                \draw[very thick] (0,0)--(0,-9)--(4,-9)--(4,-6)--(6,-6)--(6,-3)--(9,-3)--(9,-2)--(13,-2)--(13,0)--cycle;
                \node[above] (l) at (3,-3) {$\lambda$};
                \node[left] (i) at (0,-4.4) {$i$};
                \node[above] (j) at (2.5,0) {$j$};
            \end{tikzpicture}
        \end{minipage}
    \end{center}
    \caption{A Young diagram $\mu$ satisfying $\mu + \rho = (\lambda + \rho \pm m\epsilon_i)^+$ or $\mu' + \rho = (\lambda' + \rho \pm m\epsilon_j)^+$ }\label{fig:ribbon}
\end{figure}
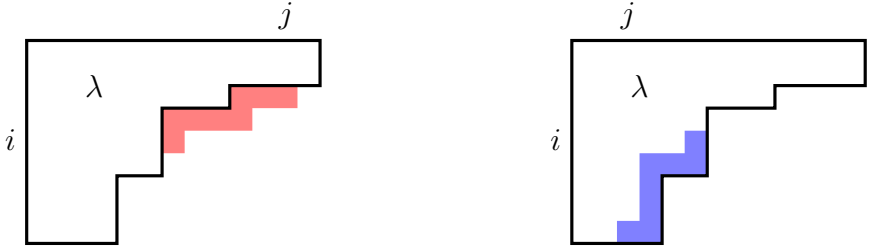

On the other hand, $b_n$ generators act diagonally on $\ket{\lambda}$:
\begin{align}
    b_{m} \ket{\lambda} =& p_{-m} (q^{-\lambda'-\rho-1/2}) \ket{\lambda}. \label{eq:b|lambda>}
\end{align}
This property reflects the fact that $f_{r-1/2}$ is an eigenfunction of $U^{m}$ in the vector representation, because $\psi^\dagger_{r} = \psi^\dagger_{f_{r-1/2}}$.

\section{Construction of the S-dual basis} \label{sec:S-dual}
In the previous section, we demonstrate that the use of S and T dual generators of $U,V$ in the fermionic realization gives two distinguished representations (horizontal/vertical) of DIM. In this section, we study a redefinition of the fermionic basis $|\lambda\rangle \rightarrow |\lambda\rangle_{\mathcal{S}}$ which transforms as a vertical basis to the horizontal $(1,0)$ generators. At the first quantized level, such a basis is given by the Fourier transformation (\ref{e:FBasis}).  The second quantized basis should be defined as the infinite wedge product of such bases with proper normalizations. We show that such a basis transforms properly with respect to the half of the generators. In order to reproduce the correct transformation for the other half, we need to redefine the generators by some sort of ``projection operator".

\subsection{Setup and results}
\subsubsection{Definition of S-dual basis}
We start constructing eigenstates of $a_m = -W[V^m]$, which are the `S-dual' of $\ket{\lambda}$.
If we impose the condition that the empty vacuum $\ket{-\infty} = \psi_{-1/2} \psi_{-3/2} \psi_{-5/2} \cdots \ket{0}$ remains unchanged under the $\mathcal{S}$-transformation, the fact $q^{r-1/2} \psi(q^{r-1/2}) = \psi_{f_{r-1/2}^{\mathcal{S}}}$ and $\ket{\lambda} \propto \psi^\dagger_{-\lambda_1+1-1/2} \psi^\dagger_{-\lambda_2+2-1/2} \cdots \ket{-\infty}$ suggests the naive definition of the `S-dual' states:
\begin{equation}
	\psi(q^{-\lambda'_1+1-1}) \psi(q^{-\lambda'_2+2-1}) \cdots \ket{-\infty}.
\end{equation}
However, this state is obviously not well-defined. Thus we regularize this expression by the boson-fermion correspondence and the normal ordering as 
\begin{equation}
	\lim_{N \to \infty} \colon e^{-\phi(q^{-\lambda'_1+1-1})} e^{-\phi(q^{-\lambda'_2+2-1})} \cdots e^{-\phi(q^{-\lambda'_N+N-1})} e^{N\hat{q}} \colon \ket{0},
\end{equation}
where we used the identification $\psi_{-1/2} \psi_{-3/2} \cdots \psi_{-(N-1/2)} \ket{0} = e^{N\hat{q}} \ket{0}$. This setup leads to the following definition of the `S-dual' basis.
\begin{defn}[S-dual basis] \label{defn:S-dual_basis}
    We define the S-dual basis as 
    \begin{gather}
    	\ket{\lambda}_{\mathcal{S}} = \mathcal{N}_\lambda \exp \left[ - \sum_{n>0} \frac{a_{-n}}{n} \sum_{l>0} q^{(-\lambda'_l+l-1)n} \right] \ket{0} = \mathcal{N}_\lambda \exp \left[ - \sum_{n>0} \frac{a_{-n}}{n} p_n (q^{-\lambda'-\rho-1/2}) \right] \ket{0}, \\
    	\lsubscript{\bra{\lambda}}{\mathcal{S}} = \mathcal{N}'_\lambda \bra{0} \exp \left[ \sum_{n>0} \frac{a_n}{n} \sum_{l>0} q^{(-\lambda_l+l)n} \right] = \mathcal{N}'_\lambda \bra{0} \exp \left[ \sum_{n>0} \frac{a_n}{n} p_n(q^{-\lambda-\rho+1/2}) \right],
    \end{gather}
    where $\mathcal{N}_\lambda$ and $\mathcal{N}'_\lambda$ are the normalizing constants given as
    \begin{gather}
        \mathcal{N}_\lambda = \left[ (-1)^{|\lambda|} q^{\frac{1}{2}\kappa_\lambda-|\lambda|} \prod_{i,j>0} (1-q^{-\lambda_i-\lambda'_j+i+j-1})^{-1} \right]^{1/2}, \label{eq:N} \\
        \mathcal{N}'_\lambda = (-1)^{|\lambda|} \left[ (-1)^{|\lambda|} q^{-\frac{1}{2}\kappa_\lambda+|\lambda|} \prod_{i,j>0} (1-q^{-\lambda_i-\lambda'_j+i+j-1})^{-1} \right]^{1/2}. \label{eq:N'}
    \end{gather}
\end{defn}
We claim that the S-dual basis satisfies\footnote{We need to regularize naively the divergent power series in the exponent which appears in the LHS.}
\begin{equation}
	{}_\mathcal{S} \langle \lambda | \mu \rangle_\mathcal{S} = \delta_{\lambda \mu}\,. \label{eq:innerproduct_S}
\end{equation}

\subsubsection{Evaluation of the action of DIM generators}
We examine the action of generators of the $(1,0)$ representation, slightly modified by introducing a parameter $v$ as follows:
\begin{gather}
    x^+_v (z) = (1-q) W[U\delta(Vv/z)] = q^{a_0} :e^{-\sum_{n \neq 0} \frac{1-q^{n}}{n} a_n (z/v)^{-n}}:, \\
    x^-_v(z) = (1-q^{-1}) W[\delta(Vv/z)U^{-1}] = q^{-a_0} :e^{\sum_{n \neq 0} \frac{1-q^{n}}{n} a_n (z/v)^{-n}}: .
\end{gather}
We note that these operators also satisfy the commutation relations of the DIM algebra. Generalizing these operators, we can introduce the representation of general currents in the DIM algebra as
\begin{gather}
    x^m_v (z) = (1-q^m)W[U^m \delta(Vv/z)] = q^{ma_0} :e^{-\sum_{n \neq 0} \frac{1-q^{mn}}{n} a_n (z/v)^{-n}}:, \\
    x^{-m}_v(z) = (1-q^{-m}) W [\delta(Vv/z)U^{-m}] = q^{-ma_0} :e^{\sum_{n \neq 0} \frac{1-q^{mn}}{n} a_n (z/v)^{-n}}: 
\end{gather}
for $m \in \mathbb{Z}$.

These operators will add/remove $m$ boxes in the Young diagram.
In order to write down the result, we need to define the two sets $\chi_{A,v;q}^m(\lambda)$, $\chi_{R,v;q}^m(\lambda) \ (m \in \mathbb{Z}_{>0})$ as
\begin{align}
    \chi_{A,v;q}^m(\lambda) &= \Set{ vq^{\lambda_i+m-i} | 1 \leq i \leq l(\lambda)+m} \backslash \Set{ vq^{\lambda_i-i} | 1 \leq i \leq l(\lambda)} \notag \\
    &= \Set{ vq^{-\lambda'_j-1+j} | 1 \leq j \leq l(\lambda')+m} \backslash \Set{ vq^{-\lambda'_j+(m-1)+j} | 1 \leq j \leq l(\lambda')}, \label{eq:chi_A} \\
    \chi_{R,v;q}^m(\lambda) &= \Set{ vq^{\lambda_i-i} | 1 \leq i \leq l(\lambda)} \backslash \Set{ vq^{\lambda_i+m-i} | 1 \leq i \leq l(\lambda)+m} \notag \\
    &= \Set{ vq^{-\lambda'_j+(m-1)+j} | 1 \leq j \leq l(\lambda')} \backslash \Set{ vq^{-\lambda'_j-1+j} | 1 \leq j \leq l(\lambda')+m}. \label{eq:chi_R}
\end{align}
The consistency of the two expressions is ensured by the fact that $\lambda+\rho$ and $-\lambda'-\rho$ are the complementary subsequences of $\mathbb{Z}+1/2$.
We note that if $\chi \in \chi_{A,v;q}^m(\lambda)$, there exists a Young diagram $\mu \in Y_A^m(\lambda)$ such that the box $x$ at the top-right corner in $\mu/\lambda$ satisfies $\chi_{x,v;q} = \chi$. Thus there is a one-to-one correspondence between $\chi_{A,v;q}^m(\lambda)$ and $Y_A^m(\lambda)$, and we can find the correspondence between $\chi_{R,v;q}^m(\lambda)$ and $Y_R^m(\lambda)$ for the same reason.

We define a generalization of Nekrasov's $Y$ function \cite{Nekrasov:2013xda}:
\begin{gather}
	\mathcal{Y}_{\lambda,v;q}^m(z) = \frac{\prod_{\chi \in \chi_{A,v;q}^m(\lambda)} (z-\chi)}{\prod_{\chi \in \chi_{R,v;q}^m(\lambda)} (z-\chi)}. \label{eq:Y_def}
\end{gather}
The properties of the function are compiled in Appendix \ref{sec:Y}.

\begin{thm} \label{thm:x_action_on_S}
    The action of $x^{\pm m}_v(z)$ on the S-dual basis is 
    \begin{align}
        x_v^m(z) \ket{\lambda}_\mathcal{S} &= \mathcal{N}_\lambda \left[ \frac{z^m}{\mathcal{Y}_{\lambda,v;q}^m(z)} \right]_- \exp \left[ - \sum_{n>0} \frac{a_{-n}}{n} \left( -(1-q^{-mn}) (z/v)^n + p_n (q^{-\lambda'-\rho-1/2}) \right) \right] \ket{0}, \\
        x_v^{-m}(z) \ket{\lambda}_\mathcal{S} &= \mathcal{N}_\lambda \left[ \frac{\mathcal{Y}_{\lambda,v;q}^m(z)}{z^m} \right]_- \exp \left[ - \sum_{n>0} \frac{a_{-n}}{n} \left( (1-q^{-mn}) (z/v)^n + p_n (q^{-\lambda'-\rho-1/2}) \right) \right] \ket{0}, \label{eq:x-m_action_on_Sket} \\
        \lsubscript{\bra{\lambda}}{\mathcal{S}} x^{m}_v(z) &= \mathcal{N}'_\lambda \left[ \frac{\mathcal{Y}_{\lambda,v^{-1};q^{-1}}^m(z^{-1})}{z^{-m}} \right]_+ \bra{0} q^{ma_0} \exp \left[ \sum_{n>0} \frac{a_n}{n} \left( -(1-q^{mn}) (z/v)^{-n} + p_n(q^{-\lambda-\rho+1/2}) \right) \right], \\
        \lsubscript{\bra{\lambda}}{\mathcal{S}} x^{-m}_v(z) &= \mathcal{N}'_\lambda \left[ \frac{z^{-m}}{\mathcal{Y}_{\lambda,v^{-1};q^{-1}}^m(z^{-1})} \right]_+ \bra{0} q^{-ma_0} \exp \left[ \sum_{n>0} \frac{a_n}{n} \left( (1-q^{mn}) (z/v)^{-n} + p_n(q^{-\lambda-\rho+1/2}) \right) \right]
    \end{align}
    for $m \in \mathbb{Z}_{>0}$.
    Here $[f(z)]_+$ and $[f(z)]_-$ denote the expansions of $f(z)$ in the neighborhood of $z=0$ and $z=\infty$, respectively.
\end{thm}

\begin{proof}
    Using the commutation relations (\ref{eq:boson_CR}) and Lemma \ref{lem:Y_alternative}, we obtain
    \begin{align}
        x_v^m(z) \ket{\lambda}_\mathcal{S} =& \mathcal{N}_\lambda \exp \left[ \sum_{n>0} \frac{1}{n} (1-q^{mn}) (z/v)^{-n} p_n (q^{-\lambda'-\rho-1/2}) \right] \notag \\
        &\times \exp \left[ - \sum_{n>0} \frac{a_{-n}}{n} \left( (1-q^{-mn}) (z/v)^n + p_n (q^{-\lambda'-\rho-1/2}) \right) \right] \ket{0} \notag \\
        =& \mathcal{N}_\lambda \prod_{l>0} \frac{1-(v/z)q^{-\lambda'_l+(m-1)+l}}{1-(v/z)q^{-\lambda'_l-1+l}} \exp \left[ - \sum_{n>0} \frac{a_{-n}}{n} \left( (1-q^{-mn}) (z/v)^n + p_n (q^{-\lambda'-\rho-1/2}) \right) \right] \ket{0} \notag \\
        =& \mathcal{N}_\lambda \left[ \frac{z^m}{\mathcal{Y}_{\lambda,v;q}^m(z)} \right]_- \exp \left[ - \sum_{n>0} \frac{a_{-n}}{n} \left( (1-q^{-mn}) (z/v)^n + p_n (q^{-\lambda'-\rho-1/2}) \right) \right] \ket{0}. \notag
    \end{align}
    Similar computation gives the other results.
\end{proof}

\begin{prop} \label{prop:b_action_on_S}
    The action of $b_m$ is
    \begin{align}
        b_{-m} \ket{\lambda}_\mathcal{S} =& \sum_{\mu \in Y^m_A(\lambda)} (-1)^{\mathrm{ht}(\mu'/\lambda')} \ket{\mu}_\mathcal{S}, \\
        \lsubscript{\bra{\lambda}}{\mathcal{S}} b_m =& \sum_{\mu \in Y^m_A(\lambda)} (-1)^{\mathrm{ht}(\mu'/\lambda')} \lsubscript{\bra{\mu}}{\mathcal{S}} q^{-ma_0}
    \end{align}
    for $m>0$.
\end{prop}

\begin{proof}
    Applying Propositions \ref{prop:Y_decomp}, \ref{prop:N_add}, \ref{prop:N_remove} and (\ref{eq:Delta1}) to Theorem \ref{thm:x_action_on_S} and taking the coefficient of $z^0$, we obtain the result.
\end{proof}

\begin{prop} \label{prop:a_action_on_S}
    The action of $a_m$ is
    \begin{align}
        a_m \ket{\lambda}_\mathcal{S} &= - p_m(q^{-\lambda'-\rho-1/2}) \ket{\lambda}_\mathcal{S}, \\
        \lsubscript{\bra{\lambda}}{\mathcal{S}} a_{-m} &= p_m(q^{-\lambda-\rho+1/2}) \lsubscript{\bra{\lambda}}{\mathcal{S}}
    \end{align}
    for $m>0$.
\end{prop}

\begin{proof}
    This result follows from the definition of the S-dual basis (Definition \ref{defn:S-dual_basis}) and the commutation relation (\ref{eq:boson_CR}).
\end{proof}

The results in Propositions \ref{prop:b_action_on_S}, \ref{prop:a_action_on_S} agree with the vertical representations on the fermionic basis (\ref{eq:a-_action}), (\ref{eq:b|lambda>}) with $a$ and $b$ oscillators interchanged. However, we do not reproduce the other half (actions of $b_m$, $a_{-m}$ with $m>0$). We note that $b_m$ oscillators do not have central charge (they commute with each other).  Here they need play the role of $a$ oscillators which have nonvanishing central charge.  In this sense, we need to introduce a minor modification of the definition of the generators.

\subsection{Redefinition of generators by projection} \label{sec:modification}

Theorem \ref{thm:x_action_on_S} shows that the action of $x^m_v(z)$ on the S-dual basis differs from the vertical representation. For example, in the $m=1$ case we obtain
\begin{align}
	x^+_v(z) \ket{\lambda}_\mathcal{S} =& (1-q) \sum_{x \in A(\lambda)} \frac{1}{1-\chi_x/z} \mathcal{N}_{\lambda+x} \exp \left[ - \sum_{n>0} \frac{a_{-n}}{n} \left( -(1-q^{-n})(z/v)^n + \sum_{l>0} q^{(-\lambda'_l+l-1)n} \right) \right] \ket{0} \notag \\
    =& (1-q) \sum_{x \in A(\lambda)} \frac{1}{1-\chi_x/z} \ket{\lambda+x}_\mathcal{S} + (\text{positive power of } z), \label{eq:x+_action_on_Sket} \\
    x^-_v(z) \ket{\lambda}_\mathcal{S} =& \left[ (1-q^{-1}) \sum_{x \in R(\lambda)} \frac{1}{1-\chi_{x}/z} \mathcal{N}_{\lambda-x} - \frac{v}{z} \mathcal{N}_{\lambda} + (1-q^{-1}) p_1(q^{\lambda'+\rho+1/2}) \right] \notag \\
    &\times \exp \left[ - \sum_{n>0} \frac{a_{-n}}{n} \left( (1-q^{-n})(z/v)^n + \sum_{l>0} q^{(-\lambda'_l+l-1)n} \right) \right] \ket{0} \notag  \\
    =& (1-q^{-1}) \sum_{x \in R(\lambda)} \frac{1}{1-\chi_{x}/z} \ket{\lambda-x}_\mathcal{S} - \frac{v}{z} \ket{\lambda}_\mathcal{S} + (\text{non-negative power of } z), \label{eq:x-_action_on_Sket} \\
    \lsubscript{\bra{\lambda}}{\mathcal{S}} x^+_v(z) =& \left[ (1-q) \sum_{x \in R(\lambda)} \frac{1}{1-z/\chi_x} \mathcal{N}'_{\lambda-x} - \frac{z}{v} \mathcal{N}'_\lambda + (1-q) p_1(q^{-\lambda'-\rho-1/2}) \mathcal{N}'_\lambda \right] \notag \\
    & \times \bra{0} q^{a_0} \exp \left[ \sum_{n>0} \frac{a_n}{n} \left( -(1-q^n)(z/v)^{-n} + \sum_{l>0} q^{(-\lambda_l+l)n} \right) \right] \notag \\
    =& (1-q) \sum_{x \in R(\lambda)} \frac{1}{1-z/\chi_x} \lsubscript{\bra{\lambda-x}}{\mathcal{S}} q^{a_0} - \frac{z}{v} \lsubscript{\bra{\lambda}}{\mathcal{S}} q^{a_0} + (\text{non-positive power of } z), \label{eq:x+_action_on_Sbra} \\
    \lsubscript{\bra{\lambda}}{\mathcal{S}} x^-_v(z) =& (1-q^{-1})\sum_{x \in A(\lambda)} \frac{1}{1-z/\chi_x} \mathcal{N}'_{\lambda+x} \bra{0} q^{-a_0} \exp \left[ \sum_{n>0} \frac{a_n}{n} \left( (1-q^n) (z/v)^{-n} + \sum_{l>0} q^{(-\lambda_l+l)n} \right) \right]  \notag \\
    =& (1-q^{-1})\sum_{x \in A(\lambda)} \frac{1}{1-z/\chi_x} \lsubscript{\bra{\lambda+x}}{\mathcal{S}} q^{-a_0} + (\text{negative power of } z). \label{eq:x-_action_on_Sbra}
\end{align}
Especially, taking the coefficient of $z^0$, we obtain
\begin{align}
    b_1 \ket{\lambda}_\mathcal{S} &= \sum_{x \in R(\lambda)} \ket{\lambda-x}_\mathcal{S} + a_{-1} \ket{\lambda}_\mathcal{S} - p_1(q^{-\lambda-\rho+1/2}) \ket{\lambda}_\mathcal{S}, \label{eq:b1_action_on_Sket} \\
    \lsubscript{\bra{\lambda}}{\mathcal{S}} b_{-1} &= \sum_{x \in R(\lambda)} \lsubscript{\bra{\lambda-x}}{\mathcal{S}} q^{a_0} + \lsubscript{\bra{\lambda}}{\mathcal{S}} a_1 q^{a_0} + p_1(q^{-\lambda'-\rho-1/2}) \lsubscript{\bra{\lambda}}{\mathcal{S}}. \label{eq:b-1_action_on_Sbra}
\end{align}
These results show that the action of the operators which live in the shaded area in Fig. \ref{fig:DIM_gen_Sdual} on the S-dual basis is the same as $(0,1)$ representation, but other operators act differently.
\begin{figure}
    \centering
    \begin{tikzpicture}
        \fill [fill=CadetBlue1!50] (-0.75,2.0)--(-0.75,1.0)--(0.75,-0.5)--(0.75,2.0) ;
        
        \draw[->] (0,-2)--(0,2);
        \draw[->] (-1,0)--(1,0);
        \fill [Gold3] (0,1.5) circle [radius=0.8mm];
        \fill [Gold3] (0,1) circle [radius=0.8mm];
        \fill [Gold3] (0,0.5) circle [radius=0.8mm];
        \fill [Gold3] (0,0) circle [radius=0.8mm];
        \fill [Gold3] (0,-0.5) circle [radius=0.8mm];
        \fill [Gold3] (0,-1.0) circle [radius=0.8mm];
        \fill [Gold3] (0,-1.5) circle [radius=0.8mm];
        \fill [Firebrick2] (0.5,1.5) circle [radius=0.8mm];
        \fill [Firebrick2] (0.5,1) circle [radius=0.8mm];
        \fill [Firebrick2] (0.5,0.5) circle [radius=0.8mm];
        \fill [Firebrick2] (0.5,0) circle [radius=0.8mm];
        \fill [Firebrick2] (0.5,-0.5) circle [radius=0.8mm];
        \fill [Firebrick2] (0.5,-1.0) circle [radius=0.8mm];
        \fill [Firebrick2] (0.5,-1.5) circle [radius=0.8mm];
        \fill [RoyalBlue1] (-0.5,1.5) circle [radius=0.8mm];
        \fill [RoyalBlue1] (-0.5,1) circle [radius=0.8mm];
        \fill [RoyalBlue1] (-0.5,0.5) circle [radius=0.8mm];
        \fill [RoyalBlue1] (-0.5,0) circle [radius=0.8mm];
        \fill [RoyalBlue1] (-0.5,-0.5) circle [radius=0.8mm];
        \fill [RoyalBlue1] (-0.5,-1.0) circle [radius=0.8mm];
        \fill [RoyalBlue1] (-0.5,-1.5) circle [radius=0.8mm];

        \fill [Firebrick2] (3.0,1.5) circle [radius=0.8mm];
        \path (3.2,1.5) node [right] {$x^+(z)$};
        
        \fill [Gold3] (3.0,1.0) circle [radius=0.8mm];
        \path (3.2,1.0) node [right] {$\psi^\pm(z)$};
        
        \fill [RoyalBlue1] (3.0,0.5) circle [radius=0.8mm];
        \path (3.2,0.5) node [right] {$x^-(z)$};
    \end{tikzpicture}
    \caption{The operators which act the same as the $(0,1)$ representation on $\ket{\lambda}_\mathcal{S}$}
    \label{fig:DIM_gen_Sdual}
\end{figure}
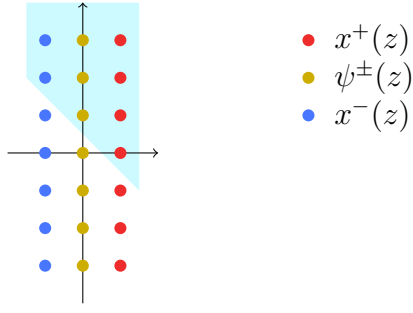

These differences arise from the following two contributions:
\begin{itemize}
    \item The differences between the singularity $\frac{1}{1-(z/\chi_x)^{\pm 1}}$ in $[\mathcal{Y}_{\lambda,v;q}^m(z)]^{\pm 1}$ and $\delta(z/\chi_x)$ in the vertical representation
    \item The poles at $z=0$ in $\frac{\mathcal{Y}_{\lambda,v;q}^m(z)}{z^m}$
\end{itemize}
The former contribution can be eliminated by inserting the projection operator
\begin{equation}
    \mathcal{P} = \sum_{\lambda} \ket{\lambda}_\mathcal{S} \lsubscript{\bra{\lambda}}{\mathcal{S}}. \label{eq:P}
\end{equation}
The latter is removed by a slight modification of the generators, as we will see later.

\begin{prop} \label{prop:PbP_action_on_S}
   We have
    \begin{gather}
        \mathcal{P} b_m \ket{\lambda}_\mathcal{S} = \sum_{\mu \in Y^m_R(\lambda)} (-1)^{\mathrm{ht}(\lambda'/\mu')} \ket{\mu}_\mathcal{S}, \\
        \lsubscript{\bra{\lambda}}{\mathcal{S}} b_{-m} \mathcal{P} = \sum_{\mu \in Y^m_R(\lambda)} (-1)^{\mathrm{ht}(\lambda'/\mu')} \lsubscript{\bra{\mu}}{\mathcal{S}}
    \end{gather}
    for $m>0$.
\end{prop}

\begin{proof}
    These follow from Proposition \ref{prop:b_action_on_S} and (\ref{eq:P}).
\end{proof}

\begin{prop} \label{prop:PaP_action_on_S}
    We have
    \begin{gather}
        \mathcal{P} a_{-m} \ket{\lambda}_\mathcal{S} = p_m(q^{-\lambda-\rho+1/2}) \ket{\lambda}_\mathcal{S}, \\
        \lsubscript{\bra{\lambda}}{\mathcal{S}} a_{m} \mathcal{P} = - p_m(q^{-\lambda'-\rho-1/2}) \lsubscript{\bra{\lambda}}{\mathcal{S}}
    \end{gather}
    for $m>0$.
\end{prop}

\begin{proof}
    We obtain the result using Theorem \ref{prop:a_action_on_S} and the definition of the projection operator (\ref{eq:P}).
\end{proof}

Thus the correspondence under the S-duality can be compiled as follows.
\begin{align}
    \ket{\lambda} \quad &\to \quad \ket{\lambda}_\mathcal{S} \\
    \bra{\lambda} \quad &\to \quad \lsubscript{\bra{\lambda}}{\mathcal{S}} \\
    a_{m} \quad &\to \quad - \mathcal{P} b_{m} \\
    a_{-m} \quad &\to \quad -b_{-m} \mathcal{P} \\
    b_{m} \quad &\to \quad - \mathcal{P} a_{-m} \\
    b_{-m} \quad &\to \quad - a_m \mathcal{P}
\end{align}
Here $m>0$.

In order to obtain the vertical representation by eliminating the poles at $z=0$ in $\frac{\mathcal{Y}_{\lambda,v;q}^m(z)}{z^m}$, we modify the generators as\footnote{We expect that $\tilde{x}^+_v(z) = \mathcal{P} x^+_v(z)$, and $\tilde{x}^-_v(z) = \mathcal{P} x^-_v(z) \mathcal{P}$ give the vertical representation on the S-dual basis. Careful evaluation, however, implies that we have additional terms proportional to $(z/v)^{\pm 1}$. } 
\begin{gather}
    \tilde{x}^+_v(z) = \mathcal{P} x^+_v(z) \mathcal{P} + \frac{z}{v}, \\
    \tilde{x}^-_v(z) = \mathcal{P} x^-_v(z) \mathcal{P} + \frac{v}{z}, \\
    \tilde{h}_m = - \frac{v^m}{m} \frac{\mathcal{P} a_m \mathcal{P}}{1-q^{-m}}, \\
    \hat{l}_1 = 0, \qquad \hat{l}_2 = 1.
\end{gather}
The modified generators give the vertical representation on the S-dual basis:
\begin{align}
    \tilde{x}_v^+(z) \ket{\lambda}_\mathcal{S} &= (1-q) \sum_{x \in A(\lambda)} \delta(z/\chi_x) \ket{\lambda+x}_\mathcal{S}, \\
    \tilde{x}_v^-(z) \ket{\lambda}_\mathcal{S} &= (1-q^{-1}) \sum_{x \in R(\lambda)} \delta(z/\chi_x) \ket{\lambda-x}_\mathcal{S}, \\
    \tilde{h}_m \ket{\lambda}_\mathcal{S} &= \begin{cases} \displaystyle \frac{v^m}{m} \frac{p_{m}(q^{-\lambda'-\rho-1/2})}{1-q^{-m}} \ket{\lambda}_\mathcal{S} \qquad (m>0) \\ \displaystyle - \frac{v^m}{m} \frac{p_{-m}(q^{-\lambda-\rho+1/2})}{1-q^{-m}} \ket{\lambda}_\mathcal{S} \qquad (m<0) \end{cases}.
\end{align}

One may interpret the role of the projection operator as follows. We take $b_1$ as an example. In terms of the standard basis, $b_m$ operators are diagonal for any $m$ (see eq.\ref{eq:b|lambda>}) and commute with each other. 
In particular, the action of $b_1$ is the multiplication of $p_{-1}$ -- the power sum polynomial of negative power. On the S-dual basis, we have to modify the oscillators to produce $[b_1,b_{-1}]=1$. 
For this purpose, the action of $b_{1}$ should be replaced by $\frac{\partial}{\partial p_1}$. This discrepancy shows up in (\ref{eq:b1_action_on_Sket}) and the projection operator removes it.
We note that similar replacement was used to construct $\mathbf{SH}$ (a degenerate algebra of DIM) in \cite{Schiffmann:2012}.  This algebra was defined as a symmetrized degenerate DAHA which operates the space of symmetric polynomials of $n$ variables and was denoted as $\mathbf{SH}_n^+$. One may define similarly the algebra for negative power symmetric polynomials $\mathbf{SH}_n^-$. We may patch these two algebras to obtain $\mathbf{SH}$ where the negative power sum is replaced by the derivative $\partial_{p_n}$. We refer the appendix B of \cite{Schiffmann:2012} for detail.

\section{Topological vertex and the amplitude} \label{s:tv}
In this section, we calculate the inner product of the standard and S-dual bases and relate it to the topological vertex when one of the indices is empty.
Then, we consider the general case and express the intertwiners (topological vertex contracted with the standard states) on the Fock basis. The inclusion of the S-dual basis makes the formula more symmetric.
We apply it to simple topological string amplitudes,  $U(1)$ and $U(2)$.  The expression implies that the S-duality transformation, which exchanges the vertical and horizontal lines in the diagram, becomes the direct consequence of the symmetry of the inner product.


\subsection{Inner product and Hopf link invariant}
First, we introduce the topological vertex \cite{Aganagic:2003db}.
\begin{defn} \label{def:tv_C}
We use the definition of the topological vertex $C_{\lambda \mu \nu}$  in terms of the skew Schur polynomials as \cite{Okounkov:2003sp}
\begin{equation} \label{eq:tv_C}
C_{\lambda \mu \nu} = q^\frac{\kappa_\mu}{2} s_{\nu'}(q^{- \rho}) \sum_\eta s_{\lambda' / \eta}(q^{-\nu -\rho}) s_{\mu / \eta}(q^{-\nu' -\rho}) ,
\end{equation}
where $\kappa_\lambda=\sum_i \lambda_i(\lambda_i-2i+1)$.
\end{defn}

It is known \cite{Aganagic:2003db} that the topological vertex has the following two properties.
\begin{eqnarray}
&&C_{\lambda\mu\nu}=C_{\mu\nu\lambda}=C_{\nu\lambda\mu}, \label{eq:tv_cyclic}
\\
&&q^{\frac{\kappa_\mu}{2}}C_{\varnothing\lambda\mu'}=q^{\frac{\kappa_\lambda}{2}}C_{\varnothing\mu\lambda'}=W_{\lambda\mu}(q). \label{eq:tv_sym}
\end{eqnarray}
In the second line, we introduced the Hopf-link invariant $W_{\lambda\mu}(q)$.

We first give the connection between the Hopf-link invariant and the inner product between the bases.
\begin{prop} \label{prop:tv_innerproduct}
The inner product of the standard and the S-dual bases is given by the Hopf-link invariant:
\begin{align}
W_{\lambda\mu}&=q^{\frac{\kappa_\lambda}{2}+\frac{\kappa_\mu}{2}}
(-q^{-\frac{1}{2}})^{|\lambda|+|\mu|}\mathcal{N}_\varnothing^{-1}\lsubscript{\braket{\lambda|\mu}}{\mathcal{S}}\label{eq:tv_innerproduct1}\\
&=q^{\frac{\kappa_\lambda}{2}+\frac{\kappa_\mu}{2}}q^{\frac{1}{2}|\lambda|+\frac{1}{2}|\mu|}\mathcal{N}_\varnothing^{-1}\braket{\mu'|\lambda'}_\mathcal{S},\label{eq:tv_innerproduct2}
\end{align}
We note that $\mathcal{N}_\varnothing=\prod_{n=1}^\infty(1-q^n)^{-n/2}$, the square root of MacMahon function, from (\ref{eq:N}). 
\end{prop}

\begin{proof}
In order to prove (\ref{eq:tv_innerproduct1}), we rewrite the components of the topological vertex by the Schur polynomials by using (\ref{eq:tv_C}), 
\begin{equation}
C_{\varnothing\lambda\mu}=q^\frac{\kappa_\lambda}{2}s_{\mu'}(q^{-\rho})s_\lambda(q^{-\mu'-\rho}) .
\end{equation}
From (\ref{eq:Schur_flip_unfinished}), we can write the Schur polynomial by the inner products as
\begin{equation}
s_{\mu'}(q^{-\lambda-\rho-\frac{1}{2}})=(-q)^{-|\mu|}\left(\mathcal{N}_{\lambda'}\mathcal{N'}_{\lambda'}\right)^{-1}
\braket{\lambda|\varnothing}_\mathcal{S}\lsubscript{\braket{\mu'|\lambda}}{\mathcal{S}} .
\end{equation}
Using this relation and $\mathcal{N'}_\lambda^{-1}\mathcal{N'}_{\lambda'}=q^\frac{\kappa_\lambda}{2}$, we can rewrite $C_{\varnothing\lambda\mu}$ as
\begin{align}
C_{\varnothing\lambda\mu}&=q^\frac{\kappa_\lambda}{2}(-q^{-\frac{1}{2}})^{|\lambda|+|\mu|}\left(\mathcal{N}_\varnothing \mathcal{N'}_\varnothing \right)^{-1}\left(\mathcal{N}_\mu\mathcal{N'}_\mu\right)^{-1}
\braket{\varnothing|\varnothing}_\mathcal{S}\lsubscript{\braket{\mu'|\varnothing}}{\mathcal{S}}
\braket{\mu'|\varnothing}_\mathcal{S}\lsubscript{\braket{\lambda|\mu'}}{\mathcal{S}} \notag\\
&=q^\frac{\kappa_\lambda}{2}(-q^{-\frac{1}{2}})^{|\lambda|+|\mu|}\mathcal{N}_\varnothing^{-1}\lsubscript{\braket{\lambda|\mu'}}{\mathcal{S}} .
\end{align}
We can derive  (\ref{eq:tv_innerproduct2}) in the same way.
\end{proof}

By using this proposition, we can express the symmetry of the two subscripts of the topological vertex as the symmetry of the inner product.

\begin{prop} \label{prop:innerproduct_sym}
We have the following properties of the inner product,
\begin{align}
\lsubscript{\braket{\lambda|\mu}}{\mathcal{S}}&=\lsubscript{\braket{\mu|\lambda}}{\mathcal{S}} ,\label{eq:sym_innerproduct1}\\
\braket{\lambda|\mu}_\mathcal{S}&=\braket{\mu|\lambda}_\mathcal{S} ,\label{eq:sym_innerproduct2}\\
\braket{\mu|\lambda}_\mathcal{S}&=(-q)^{|\lambda|+|\mu|}\lsubscript{\braket{\lambda'|\mu'}}{\mathcal{S}}.\label{eq:sym_innerproduct3}
\end{align}
for any Young diagrams $\lambda$ and $\mu$.
\end{prop}

\begin{proof}
They are immediate consequence of eqs.(\ref{eq:tv_sym},\ref{eq:tv_innerproduct1},\ref{eq:tv_innerproduct2}).
\end{proof}


\subsection{Intertwiners and the S-dual basis}
In this subsection, we go further to express the general topological string by the (dual) basis.  For that purpose, it is more convenient to work with the intertwiners ${\Phi}^*$ and ${\Phi}$ \cite{Awata:2011ce} which are obtained by contracting the topological vertex with the basis. As the name suggests, it gives the intertwiners between the representations of DIM:
\begin{eqnarray}
\Phi: \mathcal{F}^{(0,1)}_v \otimes \mathcal{F}^{(1,N)}_u \rightarrow \mathcal{F}^{(1,N+1)}_{-vu}\nonumber\\
\Phi^*: \mathcal{F}^{(1,N+1)}_{-vu}\rightarrow \mathcal{F}^{(0,1)}_v \otimes \mathcal{F}^{(1,N)}_u \nonumber
\end{eqnarray}
where $\mathcal{F}^{(l,m)}_u$ is the Fock space representation of $(l,m)$ representation with weight $u$.  They are expressed as the summation over the topological vertex whose subscript is contracted with the fermionic basis $|\lambda\rangle$ or $\langle\lambda|$ with the proper weight factors.
For simplicity, we start from giving the expression for $u=v=1$ case and denote the corresponding intertwiners as $\tilde\Phi$ and $\tilde\Phi^*$.
\begin{defn} \label{def:intertwiner}
We define the intertwiners $\tilde{\Phi}^*$ and $\tilde{\Phi}$ as
\begin{align}
\tilde{\Phi}^*&\equiv \sum_{\lambda, \mu, \nu} \left( (-1)^N q^{\frac{1}{2}}\right)^{|\lambda|}
f_\lambda^N \cdot \left(-q^{-\frac{1}{2}}\right)^{|\mu|}f_\mu^{-1} \cdot q^{{\frac{1}{2}}|\nu|}\cdot C_{\nu' \mu \lambda} \ket{\lambda} \otimes \ket{\mu} \otimes \bra{\nu},\label{eq:Phi*_def}\\
\tilde{\Phi}&\equiv \sum_{\lambda, \mu, \nu} \left(-(-1)^N q^{-\frac{1}{2}}\right)^{|\lambda|}
f_\lambda^{-N} \cdot q^{\frac{1}{2}|\mu|}f_\mu \cdot (-q^{-\frac{1}{2}})^{|\nu|}
\cdot C_{\nu \mu' \lambda'} \bra{\lambda} \otimes \bra{\mu} \otimes \ket{\nu} .\label{eq:Phi_def}
\end{align}
Here, $f_\lambda=(-1)^{|\lambda|}q^\frac{\kappa_\lambda}{2}$ is the frame factor and $C_{\lambda \mu \nu}$ is the topological vertex.
\end{defn}

\begin{defn}
We define the operator $\hat{w}$ by one set of the fermions as
\begin{equation}
\hat{w}\equiv\sum_{r\in \mathbb{Z}+\frac{1}{2}} r^2 :\psi^\dagger_r \psi_{-r}: .
\end{equation}
We note that the standard basis is its eigenstate,
\begin{equation}
\hat{w}\ket{\lambda}=\kappa_\lambda\ket{\lambda}, \label{eq:w}
\end{equation}
for any Young diagram $\lambda$.
\end{defn}

Using the above, we describe the main result in this subsection.
\begin{thm} \label{thm:int}
Use of the S-dual basis gives a symmetric representation of
the intertwiners $\tilde{\Phi}^*$ and $\tilde{\Phi}$:
\begin{align}
\tilde{\Phi}^*&=\mathcal{N}_\varnothing^{-1} \times\lsuperscript{\bra{0}}{3}\exp\left[\sum_{n>0}\frac{1}{n}\left(-(-1)^n b_n^{(1)}a_n^{(3)}+a_{-n}^{(2)}a_n^{(3)}+(-1)^n b_{-n}^{(1)}a_{-n}^{(2)}\right)\right]
q^{\frac{N}{2}\hat{w}^{(1)}}q^{L_0^{(1)}}\ket{0}^1_\mathcal{S}\ket{0}^2,\label{eq:Phi*}\\
\tilde{\Phi}&=\mathcal{N}_\varnothing^{-1} \times\lsuperscript{\bra{0}}{2}\twolscripts{\bra{0}}{1}{\mathcal{S}}(-q)^{-L_0^{(1)}}
q^{-\frac{N}{2}\hat{w}^{(1)}}\exp\left[\sum_{n>0}\frac{1}{n}\left((-1)^n b_n^{(1)}a_n^{(2)}+a_n^{(2)}a_{-n}^{(3)}-(-1)^n b_{-n}^{(1)}a_{-n}^{(3)}\right)\right]\ket{0}^3.\label{eq:Phi}
\end{align}
We use the notation: $a^{(i)}, b^{(i)}$ are the operators acting on the $i$-th Fock space. Similarly $|0\rangle^{i}, |0\rangle^{i}_\mathcal{S}$ are the vacuum of the (standard and S-dual) basis of the $i$-th Fock space. $L_0=\sum_{n>0} a_{-n} a_n$ is the degree operator.
\end{thm}
\begin{proof}
We consider $\tilde{\Phi}^*$.
Using Definition \ref{def:tv_C}, (\ref{eq:Schur}) and the definition of the frame factor, we can rewrite (\ref{eq:Phi*_def}) as
\begin{align}
\tilde{\Phi}^*&=\sum_{\lambda, \mu, \nu} q^{\frac{N}{2}\kappa_\lambda} s_{\lambda'}(q^{-\rho+\frac{1}{2}})
\sum_\eta s_{\nu / \eta}(q^{-\lambda -\rho+\frac{1}{2}}) s_{\mu / \eta}(q^{-\lambda'-\rho-\frac{1}{2}})
\ket{\lambda}\otimes \ket{\mu} \otimes \bra{\nu} \notag\\
&= \sum_{\lambda, \mu, \nu, \eta} q^{\frac{N}{2}\kappa_\lambda}\times
\lsuperscript{\bra{\lambda}}{1}\exp\left[-\sum_{n>0}\frac{p_n(q^{-\rho+\frac{1}{2}})}{n}a_{-n}\right]\ket{0}^1\times
\lsuperscript{\bra{\eta}}{3}\exp\left[\sum_{n>0}\frac{p_n(q^{-\lambda-\rho+\frac{1}{2}})}{n}(-1)^n a_n\right]\ket{\nu}^3\notag\\
&\quad\times\lsuperscript{\bra{\mu}}{2}\exp\left[\sum_{n>0}\frac{p_n(q^{-\lambda'-\rho-\frac{1}{2}})}{n}(-1)^n a_{-n}\right]\ket{\eta}^2
\times\ket{\lambda}^3 \ket{\mu}^2 \lsuperscript{\bra{\nu}}{3}. \label{eq:Phi_proof1}
\end{align}
Then, using (\ref{eq:w}) and
\begin{equation}
\sum_\eta \ket{\eta}^2 \lsuperscript{\bra{\eta}}{3}
=\lsuperscript{\bra{0}}{3}\exp\left[\sum_{n>0}\frac{1}{n}a_{-n}^{(2)}a_n^{(3)}\right]\ket{0}^2,
\end{equation}
we can rewrite (\ref{eq:Phi_proof1}) as
\begin{align}
\tilde{\Phi}^*&=\sum_{\lambda, \eta} \ket{\lambda}^1\lsuperscript{\bra{\lambda}}{1}q^{\frac{N}{2}\hat{w}}\exp\left[-\sum_{n>0}\frac{p_n(q^{-\rho+\frac{1}{2}})}{n}a_{-n}^{(1)}\right]\ket{0}^1\notag\\
&\quad\times\exp\left[\sum_{n>0}\frac{p_n(q^{-\lambda'-\rho-\frac{1}{2}})}{n}(-1)^n a_{-n}^{(2)}\right]\ket{\eta}^2
\lsuperscript{\bra{\eta}}{3}\exp\left[\sum_{n>0}\frac{p_n(q^{-\lambda-\rho+\frac{1}{2}})}{n}(-1)^n a_n^{(3)}\right]\notag\\
&=\sum_\lambda \ket{\lambda}^1\lsuperscript{\bra{\lambda}}{1}q^{\frac{N}{2}\hat{w}}\exp\left[-\sum_{n>0}\frac{p_n(q^{-\rho+\frac{1}{2}})}{n}a_{-n}^{(1)}\right]\ket{0}^1\notag\\
&\quad\times\lsuperscript{\bra{0}}{3}\exp\left[\sum_{n>0}\frac{p_n(-q^{-\lambda-\rho+\frac{1}{2}})}{n}a_n^{(3)}\right]
\exp\left[\sum_{n>0}\frac{1}{n}a_{-n}^{(2)}a_n^{(3)}\right]
\exp\left[\sum_{n>0}\frac{p_n(-q^{-\lambda'-\rho-\frac{1}{2}})}{n}a_{-n}^{(2)}\right]\ket{0}^2. \label{eq:Phi_proof2}
\end{align}
Furthermore, using
\begin{equation}
b_{-m}\ket{\lambda}=p_m(q^{-\lambda'-\rho-\frac{1}{2}})\ket{\lambda}
\end{equation}
and
\begin{equation}
b_m\ket{\lambda}=-p_m(q^{-\lambda-\rho+\frac{1}{2}})\ket{\lambda},
\end{equation}
we can rewrite (\ref{eq:Phi_proof2}) as
\begin{align}
\tilde{\Phi}^*&=\lsuperscript{\bra{0}}{3}\exp\left[\sum_{n>0}\frac{1}{n}\left(-(-1)^nb_n^{(1)}a_n^{(3)}+a_{-n}^{(2)}a_n^{(3)}+(-1)^n b_{-n}^{(1)}a_{-n}^{(2)}\right)\right]\notag\\
&\quad\times q^{\frac{N}{2}\hat{w}^{(1)}}\exp\left[-\sum_{n>0}\frac{p_n(q^{-\rho+\frac{1}{2}})}{n}a_{-n}^{(1)}\right]
\ket{0}^1\ket{0}^2.
\end{align}
By using Definition \ref{defn:S-dual_basis} and the fact that
\begin{equation}
L_0 \prod_{i=1}^m a_{-n_i}\ket{0}=\left(\sum_{i=1}^m n_i \right)\prod_{i=1}^m a_{-n_i}\ket{0},
\end{equation}
we obtain (\ref{eq:Phi*}).
We can prove (\ref{eq:Phi}) in the same way.
\end{proof}

One may easily include the parameters $u, v$ in the Fock space $\mathcal{F}^{(l,m)}$ to derive the general expression of the intertwiners.

\begin{defn} \label{def:AFSintertwiner}
General expression of the intertwiners $\Phi^*$ and $\Phi$ \cite{Awata:2011ce}:\
\begin{align}
\Phi^*&\equiv \sum_{\lambda, \mu, \nu} \left( \frac{(-u)^N}{q^{-\frac{1}{2}}v}\right)^{|\lambda|}
f_\lambda^N \cdot \left(-q^{-\frac{1}{2}}u\right)^{|\mu|}f_\mu^{-1} \cdot (q^{-\frac{1}{2}}u)^{-|\nu|}\cdot C_{\nu' \mu \lambda} \ket{\lambda} \otimes \ket{\mu} \otimes \bra{\nu},\label{eq:AFSPhi*_def}\\
\Phi&\equiv \sum_{\lambda, \mu, \nu} \left(-\frac{q^{-\frac{1}{2}}u}{(-v)^N}\right)^{|\lambda|}
f_\lambda^{-N} \cdot \left(q^{-\frac{1}{2}}v\right)^{-|\mu|}f_\mu \cdot (-q^{-\frac{1}{2}}v)^{|\nu|}
\cdot C_{\nu \mu' \lambda'} \bra{\lambda} \otimes \bra{\mu} \otimes \ket{\nu} .\label{eq:AFSPhi_def}
\end{align}
\end{defn}

For simplicity, we define the following constants.
\begin{defn}
We define the constants $\alpha, \tilde{\alpha}, \beta, \tilde{\beta}$ as
\begin{align}
\alpha &\equiv q^\frac{1}{2} u^N v^{-1} , \\
\beta &\equiv -q^{-\frac{1}{2}}u , \\
\tilde{\alpha} &\equiv q^\frac{1}{2} u^{-1}v^N , \\
\tilde{\beta} &\equiv -q^{-\frac{1}{2}}v .
\end{align}
\end{defn}

In the same way as Theorem \ref{thm:int}, we can prove the following theorem.
\begin{thm} \label{thm:AFSint}
We can express the intertwiners $\Phi^*$ and $\Phi$ by the bosonic oscillators, the standard basis and the S-dual basis as follows.
\begin{align}
\Phi^*&=\mathcal{N}_\varnothing^{-1} \times\lsuperscript{\bra{0}}{3}\exp\left[\sum_{n>0}\frac{1}{n}\left(-\beta^{-n}q^{-\frac{n}{2}}b_n^{(1)}a_n^{(3)}+a_{-n}^{(2)}a_n^{(3)}+\beta^n q^\frac{n}{2}b_{-n}^{(1)}a_{-n}^{(2)}\right)\right]\notag\\
&\quad\times q^{\frac{N}{2}\hat{w}^{(1)}}(\alpha q^\frac{1}{2})^{L_0^{(1)}}\ket{0}^1_\mathcal{S}\ket{0}^2,\label{eq:AFSPhi*}\\
\hfil\notag\\
\Phi&=\mathcal{N}_\varnothing^{-1} \times\lsuperscript{\bra{0}}{2}\twolscripts{\bra{0}}{1}{\mathcal{S}}(-\tilde{\alpha}q^\frac{1}{2})^{-L_0^{(1)}}
q^{-\frac{N}{2}\hat{w}^{(1)}}\notag\\
&\quad\times\exp\left[\sum_{n>0}\frac{1}{n}\left(\tilde{\beta}^{-n}q^{-\frac{n}{2}}b_n^{(1)}a_n^{(2)}+a_n^{(2)}a_{-n}^{(3)}-\tilde{\beta}^n q^\frac{n}{2}b_{-n}^{(1)}a_{-n}^{(3)}\right)\right]\ket{0}^3.\label{eq:AFSPhi}
\end{align}
\end{thm}

\subsection{Amplitudes}
In this subsection, we present some simple topological string amplitudes where the S-duality is a direct consequence of the symmetry of the inner product between the standard and S-dual bases.

\subsubsection{The $U(1)$ case}
\begin{figure}[h]
\centering
\begin{tabular}{cc}
\begin{minipage}{0.45\hsize}
\centering
\includegraphics[width=7cm]{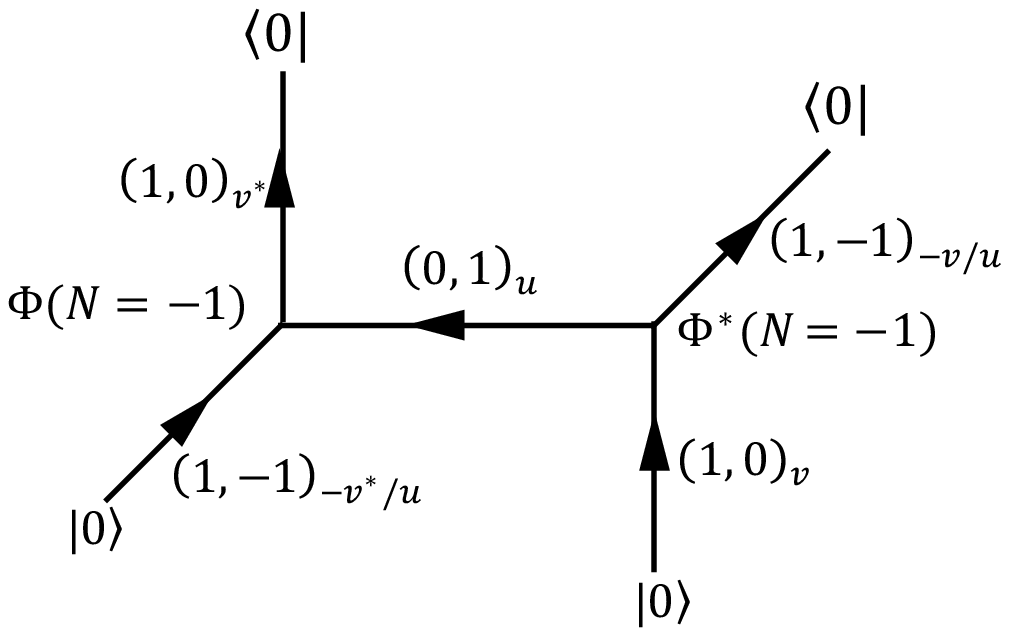}
\end{minipage}
\begin{minipage}{0.45\hsize}
\centering
\includegraphics[width=5cm]{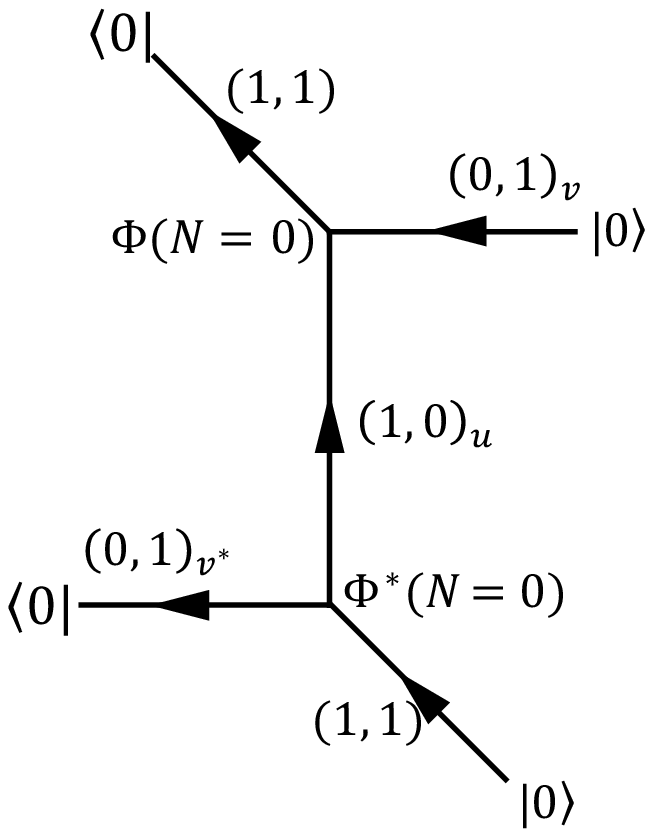}
\end{minipage}
\end{tabular}
\caption{Amplitudes of $U(1)$ and $U(1)$-dual}
\label{fig:U1ori}
\end{figure}

First, we consider the simplest case, $U(1)$, corresponding to the left of Fig. \ref{fig:U1ori}.
Using Definition \ref{def:AFSintertwiner} or Theorem \ref{thm:AFSint}, the amplitude $Z_{U(1)}$ is
\begin{align} \label{eq:ZU1}
Z_{U(1)}&=\sum_{\lambda} \bra{0} \Phi_\lambda(N=-1) \ket{0} \bra{0} \Phi_\lambda^*(N=-1) \ket{0} \notag\\
&=\sum_{\lambda}\left(-\frac{v^*}{v}\right)^{|\lambda|} C_{\varnothing \varnothing \lambda'}C_{\varnothing \varnothing \lambda} \notag\\
&=\sum_\lambda\left(-\frac{v^*}{v q}\right)^{|\lambda|}\mathcal{N}_\varnothing^{-2}
\lsubscript{\braket{\varnothing|\lambda}}{\mathcal{S}}\cdot 
\lsubscript{\braket{\varnothing|\lambda'}}{\mathcal{S}} .
\end{align}
We can also calculate the $U(1)$-dual case, corresponding to the right of Fig. \ref{fig:U1ori}.
Vertical lines and horizontal lines are exchanged in the $U(1)$ case and the $U(1)$-dual case.
The amplitude $Z'_{U(1)}$ is
\begin{align} \label{eq:ZU1dual}
Z'_{U(1)}&=\sum_{\lambda} \bra{0} \Phi_\varnothing(N=0) \ket{\lambda} \bra{\lambda} \Phi_\varnothing^*(N=0) \ket{0} \notag\\
&=\sum_{\lambda}\left(-\frac{v^*}{v}\right)^{|\lambda|} C_{\varnothing \lambda' \varnothing}C_{\varnothing \lambda \varnothing} \notag\\
&=\sum_\lambda\left(-\frac{v^*}{v q}\right)^{|\lambda|}\mathcal{N}_\varnothing^{-2}
\lsubscript{\braket{\lambda'|\varnothing}}{\mathcal{S}}\cdot
\lsubscript{\braket{\lambda|\varnothing}}{\mathcal{S}} .
\end{align}
From Proposition \ref{prop:innerproduct_sym},
\begin{equation}
Z_{U(1)}=Z'_{U(1)} .
\end{equation}
(\ref{eq:ZU1}) and (\ref{eq:ZU1dual}) suggest that the horizontal lines correspond to the standard basis $\{\ket{\lambda}\}$  and the vertical lines correspond to the S-dual basis $\{\ket{\lambda}_\mathcal{S}\}$.

\subsubsection{The $U(2)$ case}
\begin{figure}[h]
\centering
\begin{tabular}{cc}
\begin{minipage}{0.45\hsize}
\centering
\includegraphics[width=7cm]{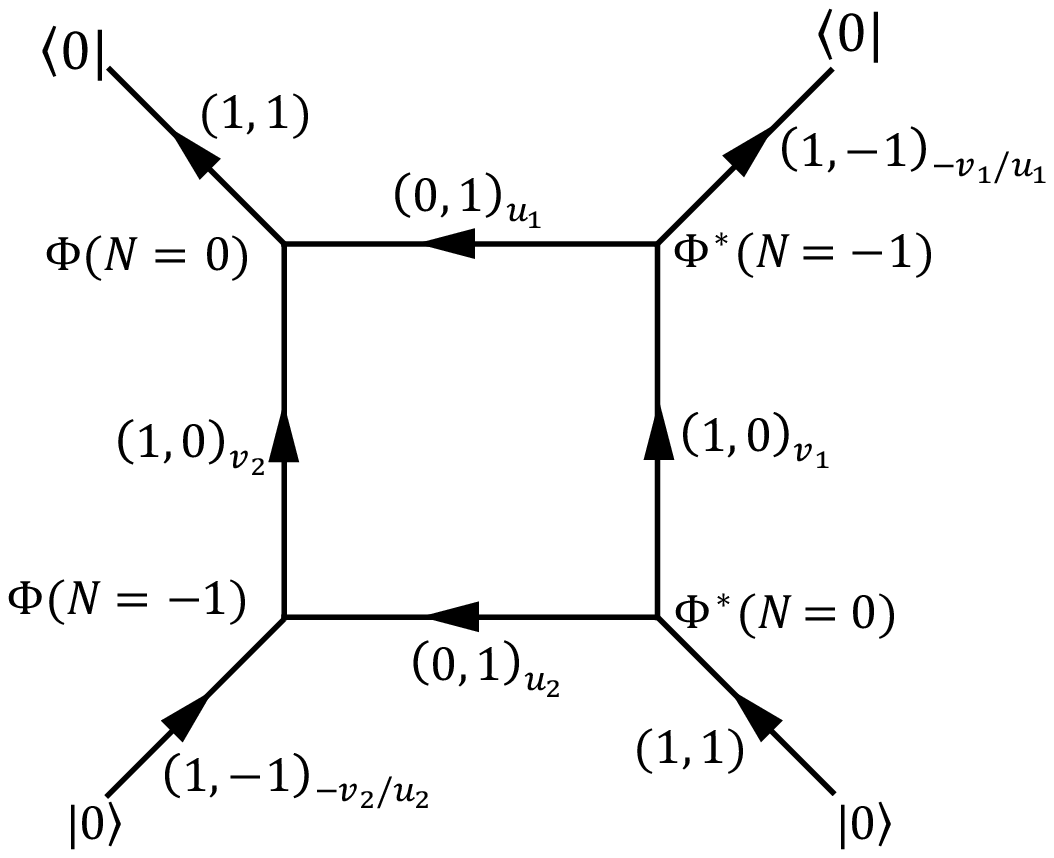}
\end{minipage}
\begin{minipage}{0.45\hsize}
\centering
\includegraphics[width=7cm]{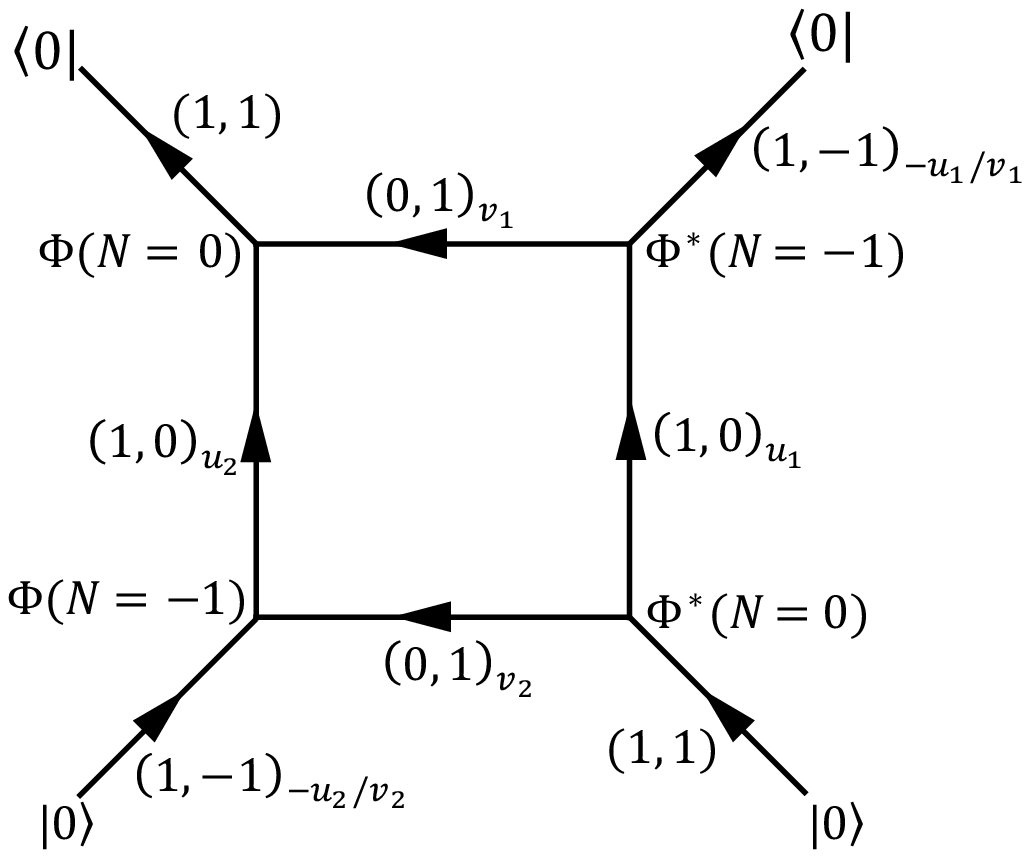}
\end{minipage}
\end{tabular}
\caption{Amplitudes of $U(2)$ and $U(2)$-dual}
\label{fig:SU2ori}
\end{figure}

We consider the more complicated case that corresponds to $U(2)$ gauge group, depicted by the left of Fig. \ref{fig:SU2ori}.
Using Definition \ref{def:AFSintertwiner} or Theorem \ref{thm:AFSint}, the amplitude $Z_{U(2)}$ is
\begin{align} \label{eq:ZSU2}
Z_{U(2)}&=\sum_{\sigma, \tau, \mu, \nu} \bra{0} \Phi_\nu(N=0)\ket{\sigma} \bra{\sigma} \Phi_\mu(N=-1) \ket{0}
\bra{0} \Phi_\nu^*(N=-1) \ket{\tau} \bra{\tau} \Phi_\mu^*(N=0) \ket{0} \notag\\
&=\sum_{\sigma, \tau, \mu, \nu}\left(\frac{v_2}{v_1}\right)^{|\mu|+|\nu|} \left(\frac{u_2}{u_1}\right)^{|\sigma|+|\tau|}
q^{\frac{\kappa_\mu}{2}-\frac{\kappa_\nu}{2}+\frac{\kappa_\sigma}{2}-\frac{\kappa_\tau}{2}}
C_{\varnothing \sigma' \nu'}C_{\sigma \varnothing \mu'}C_{\tau' \varnothing \nu}C_{\varnothing \tau \mu}\notag\\
&=\sum_{\sigma, \tau, \mu, \nu}\left(\frac{v_2}{v_1 q}\right)^{|\mu|+|\nu|} \left(\frac{u_2}{u_1 q}\right)^{|\sigma|+|\tau|}
\mathcal{N}_\varnothing^{-4}\lsubscript{\braket{\sigma|\nu}}{\mathcal{S}}\cdot\lsubscript{\braket{\sigma|\mu}}{\mathcal{S}}\cdot
\lsubscript{\braket{\tau|\nu}}{\mathcal{S}}\cdot\lsubscript{\braket{\tau|\mu}}{\mathcal{S}}.
\end{align}
Computation of the $U(2)$-dual case, corresponding to the right of Fig. \ref{fig:SU2ori} is similar.
Vertical lines and horizontal lines are exchanged in the $U(2)$ case and the $U(2)$-dual case.
The amplitude $Z'_{U(2)}$ is
\begin{equation} \label{eq:ZSU2dual}
Z'_{U(2)}
=\sum_{\sigma, \tau, \mu, \nu}\left(\frac{v_2}{v_1 q}\right)^{|\mu|+|\nu|} \left(\frac{u_2}{u_1 q}\right)^{|\sigma|+|\tau|}
\mathcal{N}_\varnothing^{-4}\lsubscript{\braket{\mu|\tau}}{\mathcal{S}}\cdot\lsubscript{\braket{\mu|\sigma}}{\mathcal{S}}\cdot
\lsubscript{\braket{\nu|\tau}}{\mathcal{S}}\cdot\lsubscript{\braket{\nu|\sigma}}{\mathcal{S}}.
\end{equation}
From Proposition \ref{prop:innerproduct_sym},
\begin{equation}
Z_{U(2)}=Z'_{U(2)} .
\end{equation}
(\ref{eq:ZSU2}) and (\ref{eq:ZSU2dual}) suggest that the horizontal lines correspond to the standard basis $\{\ket{\lambda}\}$  and the vertical lines correspond to the S-dual basis $\{\ket{\lambda}_\mathcal{S}\}$.

From the above two examples, it can be found that by introducing the S-dual basis, we can express the S-duality by the symmetry of the inner product.
As long as the amplitude is expressed by the combination of $C_{\varnothing\sigma\nu}$, similar argument can be applied to prove the S-duality of the amplitude.

\section{Conclusion} \label{sec:conclusion}

In this paper, we proposed a  new set of basis in the free fermion Fock space labeled by partitions.  It converts the $(1,0)$ representation to the $(0,1)$ representation, which realizes the S-duality, or Miki automorphism.  We argued that the straightforward change of the basis does not produce the proper change of the representation, and we need to introduce the projection operator. We note that a similar operation was used in mathematical literature such as \cite{Schiffmann:2012} to connect the positive and negative modes. Finally, we show that the S-duality of some topological string amplitudes is the direct consequence of the symmetry of the inner product between the standard and S-dual bases.

There are many things which should be studied in the future.
\begin{itemize}
\item In this paper, we restrict ourselves to the free fermionic system to explore the property of the S-duality basis.  Since the S-duality is already an established property of the DIM algebra, one should obtain a similar result for more general parameter choice.
\item In Section \ref{sec:modification} , we note that we need the extra modification of the generators and the shift by $(z/v)^{\pm 1}$ to obtain the proper $(0,1)$  representation.  It may be a manifestation of a quantum anomaly, but we do not know its exact nature at this moment.
\item In principle, there should be a new basis for each generator of $SL(2,\mathbb{Z})$.
\end{itemize}

\section*{Acknowledgement}
We would like to thank J.-E. Bourgine, O. Foda, M. Fukuda, R.-D. Zhu and K. Harada for the valuable discussions. The research of AW was partially supported by the Program for Leading Graduate Schools, MEXT, Japan. 
The research of YM is partially supported by
Grant-in-Aid MEXT/JSPS KAKENHI 18K03610.

\appendix

\section{Formulas concerning $\mathcal{Y}_{\lambda,v;q}^m(z)$, $\mathcal{N}_\lambda$ and $\mathcal{N}'_\lambda$} \label{sec:Y}

\begin{lem} \label{lem:diff_A_R}
    The size of sets $\chi_{A,v;q}^m(\lambda)$ and $\chi_{R,v;q}^m(\lambda)$ satisfies
    \begin{equation}
        |\chi_{A,v;q}^m(\lambda)| - |\chi_{R,v;q}^m(\lambda)| = m.
    \end{equation}
\end{lem}

\begin{proof}
    From (\ref{eq:chi_A}) and (\ref{eq:chi_R}), we can derive
    \begin{align}
        |\chi_{A,v;q}^m(\lambda)| - |\chi_{R,v;q}^m(\lambda)| &= \left| \Set{ vq^{\lambda_l+m-l} | 1 \leq l \leq l(\lambda)+m } \right| - \left| \Set{ vq^{\lambda_l-l} | 1 \leq l \leq l(\lambda) } \right| \notag \\
        &= (l(\lambda) + m) - l(\lambda) = m.
    \end{align}
\end{proof}

\begin{lem} \label{lem:Y_alternative}
    If $z \neq 0$, then $\mathcal{Y}_{\lambda,v;q}^m(z)$ can be written as
    \begin{equation}
        \mathcal{Y}_{\lambda,v;q}^m(z) = \begin{cases} \displaystyle z^m \prod_{l>0} \frac{z-vq^{\lambda_l+m-l}}{z-vq^{\lambda_l-l}} \qquad (|q|>1) \\ \displaystyle z^m \prod_{l>0} \frac{z-vq^{-\lambda'_l-1+l}}{z-vq^{-\lambda'_l+(m-1)+l}} \qquad (|q|<1)  \end{cases}. \label{eq:Y_alternative}
    \end{equation}
\end{lem}

\begin{proof}
    From (\ref{eq:chi_A}) and (\ref{eq:chi_R}), we can compute as
    \begin{align}
        \prod_{l=1}^N \frac{z-vq^{\lambda_l+m-l}}{z-vq^{\lambda_l-l}} =& \frac{\prod_{l=1}^{l(\lambda)+m} (z-vq^{\lambda_l+m-l})}{\prod_{l=1}^{l(\lambda)} (z-vq^{\lambda_l-l})} \frac{\prod_{l=l(\lambda)+m}^N (z-vq^{m-l})}{\prod_{l=l(\lambda)}^N (z-vq^{-l})} \notag \\
        =& \frac{\prod_{\chi \in \chi_{A,v;q}(\lambda)} (z-\chi)}{\prod_{\chi \in \chi_{R,v;q}(\lambda)} (z-\chi)} \prod_{l=N-m+1}^N \frac{1}{z-vq^{-l}}.
    \end{align}
    Taking the limit $N \to \infty$ we obtain the identity in the $|q|>1$ case.
    We can confirm the $|q|<1$ case by a similar computation.
\end{proof}

\begin{cor} \label{cor:Y_recursion}
    If a box $x$ can be added to $\lambda$, then
    \begin{equation}
        \mathcal{Y}^m_{\lambda+x,v;q}(z) = \frac{(z-q^{-1}\chi_x)(z-q^{m}\chi_x)}{(z-q^{m-1}\chi_x)(z-\chi_x)} \mathcal{Y}^m_{\lambda,v;q}(z).
    \end{equation}
    Especially,
    \begin{equation} \label{eq:Y(0)}
        \mathcal{Y}^m_{\lambda,v;q}(0) = \mathcal{Y}^m_{\varnothing,v;q}(0) = (-1)^m v^m q^{\frac{1}{2}m(m-1)}.
    \end{equation}
\end{cor}

\begin{prop} \label{prop:Y_decomp}
    $\mathcal{Y}^m_{\lambda+x,v;q}(z)$ can be decomposed as
    \begin{align}
        \frac{z^m}{\mathcal{Y}^m_{\lambda,v;q}(z)} &= \sum_{\chi \in \chi_{A,v;q}^m(\lambda)} \frac{1}{1-\chi/z} \underset{z=\chi}{\mathrm{Res}} \frac{z^{m-1}}{\mathcal{Y}^m_{\lambda,v;q}(z)} \\
        \frac{\mathcal{Y}^m_{\lambda,v;q}(z)}{z^m} &= \sum_{\chi \in \chi_{R,v;q}^m(\lambda)} \frac{1}{1-\chi/z} \underset{z=\chi}{\mathrm{Res}} \frac{\mathcal{Y}^m_{\lambda,v;q}(z)}{z^{m+1}} + \sum_{n=0}^{m} \frac{\Delta^m_{n,v;q} (\lambda)}{z^{m-n}},
    \end{align}
    where
    \begin{equation}
        \Delta_{n,v;q}^m (\lambda) = \frac{1}{n!} \frac{d^n}{dz^n} \mathcal{Y}^m_{\lambda,v;q}(0).
    \end{equation}
\end{prop}

\begin{proof}
    The meromorphic function
    \begin{equation}
        \frac{z^{m-1}}{\mathcal{Y}^m_{\lambda,v;q}(z)} = \frac{z^{m-1} \prod_{\chi \in \chi_{R,v;q}^m(\lambda)} (z-\chi)}{\prod_{\chi \in \chi_{A,v;q}^m(\lambda)} (z-\chi)}
    \end{equation}
    has simple poles at $z = \chi \in \chi_{A,v;q}^m(\lambda)$ and approaches to 0 as $z \to \infty$, which is the consequence of Lemma \ref{lem:diff_A_R}. Thus the Liouville theorem ensures that
    \begin{equation}
        \frac{z^{m-1}}{\mathcal{Y}^m_{\lambda,v;q}(z)} = \sum_{\chi \in \chi_{A,v;q}^m(\lambda)} \frac{1}{z-\chi} \underset{z=\chi}{\mathrm{Res}} \frac{z^{m-1}}{\mathcal{Y}^m_{\lambda,v;q}(z)}.
    \end{equation}
    On the other hand, the meromorphic function
    \begin{equation}
        \frac{\mathcal{Y}^m_{\lambda,v;q}(z)}{z^{m+1}} = \frac{\prod_{\chi \in \chi_{A,v;q}^m(\lambda)} (z-\chi)}{z^{m+1} \prod_{\chi \in \chi_{R,v;q}^m(\lambda)} (z-\chi)}
    \end{equation}
    approaches to 0 as $z \to \infty$ and has simple poles at $z = \chi \in \chi_{R,v;q}^m(\lambda)$ and a pole of order $m+1$ at $z=0$. The behavior in the neighborhood of $z=0$ can be determined by standard methods in the function analysis, and we obtain
    \begin{equation}
        \frac{\mathcal{Y}^m_{\lambda,v;q}(z)}{z^{m+1}} = \sum_{\chi \in \chi_{R,v;q}^m(\lambda)} \frac{1}{z-\chi} \underset{z=\chi}{\mathrm{Res}} \frac{\mathcal{Y}^m_{\lambda,v;q}(z)}{z^{m+1}} + \sum_{n=0}^m \frac{\Delta_{n,v;q}^m (\lambda)}{z^{m+1-n}}.
    \end{equation}
\end{proof}

We note that direct computation gives the explicit form of $\Delta_{n,v;q}^m (\lambda)$ as
\begin{align}
    \Delta_{0,v;q}^m (\lambda) =& (-1)^m v^m q^{\frac{1}{2}m(m-1)}, \label{eq:Delta0} \\
    \Delta_{1,v;q}^m (\lambda) =& (-1)^{m-1} v^m q^{\frac{1}{2}m(m-1)} \left( \sum_{\chi \in \chi_{A,v;q}^m(\lambda)} \frac{1}{\chi} - \sum_{\chi \in \chi_{R,v;q}^m(\lambda)} \frac{1}{\chi} \right) \notag \\
    =& \begin{cases} (-1)^m v^{m-1} q^{\frac{1}{2}m(m-1)} (1-q^{-m}) p_1(q^{-\lambda-\rho+1/2}) \qquad (|q|<1) \\
     (-1)^{m-1} v^{m-1} q^{\frac{1}{2}m(m-1)} (1-q^{-m}) p_1(q^{\lambda'+\rho+1/2}) \qquad (|q|>1) \end{cases}. \label{eq:Delta1}
\end{align}

\begin{prop} \label{prop:N_add}
    If $\mu \in Y^m_A(\lambda)$, then
    \begin{align}
        \mathcal{N}_\mu &= (-1)^{\mathrm{ht}(\mu'/\lambda')} \frac{1}{1-q^m} \underset{z=\chi_{x,v;q}}{\mathrm{Res}} \frac{z^{m-1}}{\mathcal{Y}^m_{\lambda,v;q}(z)} \mathcal{N}_\lambda, \label{eq:N_add} \\
        \mathcal{N}'_\mu &= (-1)^{\mathrm{ht}(\mu'/\lambda')} \frac{1}{1-q^{-m}} \underset{z=\chi_{x,v^{-1};q^{-1}}}{\mathrm{Res}} \frac{z^{m-1}}{\mathcal{Y}^m_{\lambda,v^{-1};q^{-1}}(z)} \mathcal{N}'_\lambda, \label{eq:N'_add}
    \end{align}
    where $x$ is the box at the top-right corner of $\mu/\lambda$.
\end{prop}

\begin{proof}
    If $\mu \in Y^m_A(\lambda)$, then there is a positive integer $j$ which satisfies $\mu' + \rho = (\lambda' + \rho + m \epsilon_j)^+$ and $\chi_{x,v;q} = vq^{-\lambda'_j-1+j}$ for the box $x$ at the top-right corner of $\mu/\lambda$ (see Fig. \ref{fig:ribbon}). Using Lemma \ref{lem:Y_alternative}, we obtain 
    \begin{equation}
        \frac{1}{1-q^m} \underset{z=\chi_{x,v;q}}{\mathrm{Res}} \frac{z^{m-1}}{\mathcal{Y}^m_{\lambda,v;q}(z)} = \prod_{l (\neq j)} \frac{1-q^{\lambda'_j-\lambda'_l-j+l+m}}{1-q^{\lambda'_j-\lambda'_l-j+l}}.
    \end{equation}
    On the other hand, direct computation gives 
    \begin{align}
        &\prod_{k,l>0} (1-q^{-\mu_k-\mu'_l+k+l-1})^{-1} \notag \\
        =& \prod_{k,l>0} (1-q^{-\lambda_k-\lambda'_l+k+l-1})^{-1} \times (-1)^m q^{m+\frac{1}{2}\kappa_\lambda-\frac{1}{2}\kappa_\mu} \left( \prod_{l (\neq j)} \frac{1-q^{\lambda'_j-\lambda'_l-j+l+m}}{1-q^{\lambda'_j-\lambda'_l-j+l}} \right)^2.
    \end{align}
    We note that the sign of the factor $\displaystyle \prod_{l (\neq j)} \frac{1-q^{\lambda'_j-\lambda'_l-j+l+m}}{1-q^{\lambda'_j-\lambda'_l-j+l}}$ is $(-1)^{\mathrm{ht}(\mu'/\lambda')}$ when $|q|<1$. Substituting these result into the definition of $\mathcal{N}_\lambda$ (\ref{eq:N}), we obtain (\ref{eq:N_add}). The proof of (\ref{eq:N'_add}) is parallel.
\end{proof}

\begin{prop} \label{prop:N_remove}
    If $\mu \in Y^m_R(\lambda)$, then
    \begin{align}
        \mathcal{N}_\mu &= (-1)^{\mathrm{ht}(\lambda'/\mu')} \frac{1}{1-q^{-m}} \underset{z=\chi_{x,v;q}}{\mathrm{Res}} \frac{\mathcal{Y}^m_{\lambda,v;q}(z)}{z^{m+1}} \mathcal{N}_\lambda, \label{eq:N_remove} \\
        \mathcal{N}'_\mu &= (-1)^{\mathrm{ht}(\lambda'/\mu')} \frac{1}{1-q^m} \underset{z=\chi_{x,v^{-1};q^{-1}}}{\mathrm{Res}} \frac{\mathcal{Y}^m_{\lambda,v^{-1};q^{-1}}(z)}{z^{m+1}} \mathcal{N}'_\lambda, \label{eq:N'_remove}
    \end{align}
    where $x$ is the box at the top-right corner of $\lambda/\mu$.
\end{prop}

\begin{proof}
    Applying the identity
    \begin{equation}
        \frac{1}{1-q^m} \underset{z=\chi_{x,v;q}}{\mathrm{Res}} \frac{z^{m-1}}{\mathcal{Y}^m_{\mu,v;q}(z)} = \left[ \frac{1}{1-q^{-m}} \underset{z=\chi_{x,v;q}}{\mathrm{Res}} \frac{\mathcal{Y}^m_{\lambda,v;q}(z)}{z^{m+1}} \right]^{-1}
    \end{equation}
    to Proposition \ref{prop:N_add}, we can derive the result.
\end{proof}

\begin{prop}
    The skew Schur polynomials can be written as
    \begin{align}
        s_{\lambda/\nu} &= \bra{\lambda} \exp \left( \sum_{n=1}^\infty \frac{p_n}{n} (-1)^n a_{-n} \right) \ket{\nu} = \bra{\nu} \exp \left( \sum_{n=1}^\infty \frac{p_n}{n} (-1)^n a_{n} \right) \ket{\lambda} \notag \\
        &= \bra{\lambda'} \exp \left( - \sum_{n=1}^\infty \frac{p_n}{n} a_{-n} \right) \ket{\nu'} = \bra{\nu'} \exp \left( - \sum_{n=1}^\infty \frac{p_n}{n} a_{n} \right) \ket{\lambda'}. \label{eq:Schur}
    \end{align}
    and
    \begin{align}
        s_{\lambda/\nu} &= \lsubscript{\bra{\lambda}}{\mathcal{S}} \exp \left( \sum_{n=1}^\infty \frac{p_n}{n} (-1)^{n-1} b_{-n} \right) \ket{\nu}_\mathcal{S} = \lsubscript{\bra{\nu}}{\mathcal{S}} \exp \left( \sum_{n=1}^\infty \frac{p_n}{n} (-1)^{n-1} b_{n} \right) \ket{\lambda}_\mathcal{S} \notag \\
        &= \lsubscript{\bra{\lambda'}}{\mathcal{S}} \exp \left( \sum_{n=1}^\infty \frac{p_n}{n} b_{-n} \right) \ket{\nu'}_\mathcal{S} = \lsubscript{\bra{\nu'}}{\mathcal{S}} \exp \left( \sum_{n=1}^\infty \frac{p_n}{n} b_{n} \right) \ket{\lambda'}_\mathcal{S}. \label{eq:Schur_S}
    \end{align}
\end{prop}

\begin{proof}
    (\ref{eq:Schur_S}) is obtained by applying
    \begin{gather}
        \mathcal{P} b_{-n_1} b_{-n_2} \cdots b_{-n_i} \ket{\lambda}_\mathcal{S} = b_{-n_1} b_{-n_2} \cdots b_{-n_i} \ket{\lambda}_\mathcal{S}, \\
        \lsubscript{\bra{\lambda}}{\mathcal{S}} b_{n_1} b_{n_2} \cdots b_{n_i} \mathcal{P} = \lsubscript{\bra{\lambda}}{\mathcal{S}} b_{n_1} b_{n_2} \cdots b_{n_i}
    \end{gather}
    to (\ref{eq:Schur}).
\end{proof}

\begin{prop} \label{prop:Schur_flip}
    \begin{align}
        \mathcal{N}_{\lambda} s_{\mu'} (q^{-\lambda-\rho-1/2}) &= q^{\frac{1}{2}\kappa_\lambda } \mathcal{N}_{\mu} s_\lambda (q^{-\mu'-\rho-1/2}), \\
        \mathcal{N}_{\lambda'} s_{\mu'} (q^{-\lambda-\rho-1/2}) &= q^{\frac{1}{2}\kappa_\mu } \mathcal{N}_{\mu'} s_\lambda (q^{-\mu'-\rho-1/2}), \label{eq:Schur_flip_1} \\
        \mathcal{N}'_{\lambda} s_{\mu'} (q^{-\lambda-\rho+1/2}) &= (-1)^{|\lambda|+|\mu|}  q^{\frac{1}{2}\kappa_\mu } \mathcal{N}'_{\mu} s_\lambda (q^{-\mu'-\rho+1/2}), \\
        \mathcal{N}'_{\lambda'} s_{\mu'} (q^{-\lambda-\rho+1/2}) &= (-1)^{|\lambda|+|\mu|}  q^{\frac{1}{2}\kappa_\lambda } \mathcal{N}'_{\mu'} s_\lambda (q^{-\mu'-\rho+1/2}).
    \end{align}
\end{prop}

\begin{proof}
    The following equality follows from (\ref{eq:Schur_S}):
    \begin{align}
        s_{\mu'}(q^{-\lambda-\rho-1/2}) =& (-q)^{-|\mu|} \lsubscript{\bra{\mu'}}{\mathcal{S}} \exp \left[ - \sum_{n>0} \frac{1}{n} p_n(q^{-\lambda-\rho+1/2}) b_{-n} \right] \sum_{\nu} \ket{\nu} \braket{\nu | \varnothing}_\mathcal{S}.
    \end{align}
    Using (\ref{eq:b|lambda>}) and 
    \begin{equation}
        \exp \left( - \sum_{n>0} \frac{1}{n} p_n(q^{-\lambda-\rho+1/2}) p_n(q^{-\nu'-\rho-1/2}) \right) = \prod_{i,j>0} (1-q^{-\lambda_i-\nu'_j+i+j-1}) = \delta_{\lambda \nu} (\mathcal{N}_{\lambda'} \mathcal{N}'_{\lambda'})^{-1},
    \end{equation}
    we obtain
    \begin{align}
        s_{\mu'}(q^{-\lambda-\rho-1/2}) =& (-q)^{-|\mu|} (\mathcal{N}_{\lambda'} \mathcal{N}'_{\lambda'})^{-1} \lsubscript{\braket{\mu' | \lambda}}{\mathcal{S}} \braket{\lambda | \varnothing}_\mathcal{S} \notag \\
        =& (-1)^{|\lambda|+|\mu|} q^{-|\mu|} \mathcal{N}_\varnothing \mathcal{N}'_{\mu'} (\mathcal{N}_{\lambda'} \mathcal{N}'_{\lambda'})^{-1} s_\lambda(q^{-\mu'-\rho+1/2}) s_{\lambda'} (q^{-\rho-1/2}) \label{eq:Schur_flip_unfinished},
    \end{align}
    where we used (\ref{eq:Schur_S}) again. Replacing $\mu \to \lambda$ and $\lambda \to \varnothing$ in this equation, one can derive
    \begin{equation}
        s_{\lambda'} (q^{-\rho-1/2}) = (-q)^{-|\lambda|} \mathcal{N}'_{\lambda'} \mathcal{N}^{-1}_\varnothing.
    \end{equation}
    Substituting this into (\ref{eq:Schur_flip_unfinished}), we can prove (\ref{eq:Schur_flip_1}). The other equations can be confirmed by using the definition of the normalizing constant (\ref{eq:N})(\ref{eq:N'}).
\end{proof}

\section{Symmetric functions}
It is known that there is one-to-one correspondence between Fock states and symmetric functions:
\begin{equation}
    a_{-\lambda_1} a_{-\lambda_2} \cdots a_{-\lambda_n} \ket{0} \ \leftrightarrow \ p_{\lambda_1} p_{\lambda_2} \cdots p_{\lambda_n}.
\end{equation}
We note that this identification can be realized by acting
\begin{equation}
    \bra{0} \exp \left( \sum_{n>0} \frac{a_n}{n} p_n \right)
\end{equation}
on Fock states. In this identification, $\ket{\lambda}$ is mapped to the Schur function as
\begin{equation}
    s_\lambda = \bra{0} \exp \left( \sum_{n>0} \frac{a_n}{n} p_n \right) \ket{\lambda}. \label{eq:Schur_free_field}
\end{equation}


Using this free field realization of the Schur polynomials, we can confirm 
\begin{align}
    \langle \lambda | \mu \rangle_\mathcal{S} &= \mathcal{N}_{\mu} s_{\lambda'}(q^{-\mu'-\rho-1/2}) = q^{\frac{1}{2} \kappa_\lambda} \mathcal{N}_{\lambda'} s_{\mu'} (q^{-\lambda'-\rho-1/2}), \\
    {}_\mathcal{S} \langle \mu | \lambda \rangle &= (-1)^{|\lambda|} \mathcal{N}'_{\mu} s_\lambda(q^{-\mu-\rho+1/2}) = (-1)^{|\mu|} q^{ - \frac{1}{2} \kappa_\lambda} \mathcal{N}'_{\lambda'} s_{\mu} (q^{-\lambda-\rho+1/2}).
\end{align}
The second equality in each equation follows from Proposition \ref{prop:Schur_flip}. Thus the S-dual states are expanded as
\begin{equation}
    \ket{\lambda}_\mathcal{S} = \sum_\mu q^{\frac{1}{2} \kappa_\mu} \mathcal{N}_{\mu'} s_{\lambda'}(q^{-\mu'-\rho-1/2}) \ket{\mu}.
\end{equation}
By this equation and (\ref{eq:b|lambda>}), we obtain
\begin{align}
    b_m \ket{\lambda'}_\mathcal{S} &= \sum_\mu q^{\frac{1}{2} \kappa_\mu} \mathcal{N}_{\mu'} (\textcolor{blue}{p_{-m}} s_{\lambda})(q^{-\mu'-\rho-1/2}) \ket{\mu}. \label{eq:b_on_Sdual}
\end{align}
Naively, we can identify the action of $b_m$ on $\ket{\lambda'}_\mathcal{S}$ as a multiplication  of $s_\lambda$ by $p_{-m}$.
However, this interpretation does not work because $p_{-m} s_\lambda \ (m>0)$ is not a polynomial and ill-defined. Thus the insertion of the projection operator $\mathcal{P}$ corresponds to an elimination of the negative-power part.

\bibliographystyle{utphys}
\bibliography{bib}

\end{document}